\documentclass[11pt,twoside]{article}
\usepackage[english]{babel}
\usepackage[utf8]{inputenc} 
\usepackage{amsthm,amsmath}
\usepackage{amssymb}
\usepackage[table]{xcolor} 
\usepackage{graphicx}
\usepackage[ruled,vlined,linesnumbered]{algorithm2e}
\usepackage{multicol}
\usepackage{multirow}
\usepackage{hyperref}
\usepackage{cite}
\usepackage{tikz}
\usepackage[small,bf]{caption}
\usepackage{enumerate}
\usepackage{pgflibraryshapes}
\usepackage[top=1in, bottom=1in, left=1in, right=1in]{geometry}
\usepackage[obeyDraft]{todonotes}
\usepackage{times}
\usepackage[T1]{fontenc}
\usepackage{setspace}

\usetikzlibrary{decorations.pathreplacing}

\makeatletter
\numberwithin{equation}{section}

\newtheorem{theorem}{Theorem}[section] 
\newtheorem{proposition}{Proposition} 
\newtheorem{lemma}{Lemma}

\newtheorem{definition}{Definition}
\newtheorem{remark}{Remark}
\newtheorem{problem}{Problem}

\newtheorem{corollary}{Corollary}
\newtheorem{hypothesis}{Hypothesis}

\newenvironment{bezout}{\vspace*{3mm}\noindent\textbf{\textit{Bezout's bound:} \itshape}}{\vspace*{3mm}}

\newcommand{\N}{\mathbb{N}}

\newcommand{\Q}{\mathbb{Q}}

\newcommand{\F}{\mathbb{F}}
\newcommand{\G}{\mathbb{G}}
\newcommand{\K}{\mathbb{K}}
\newcommand{\A}{\mathbb{A}}

\definecolor{vert}{rgb}{0,0.5,0}

\newcommand{\glex}{\G_{\operatorname{lex}}}
\newcommand{\gdrl}{\G_{\operatorname{drl}}}
\newcommand{\id}[1]{\mathcal{#1}}
\newcommand{\Id}[1]{\left \langle #1 \right \rangle}
\newcommand{\inId}[2]{\text{in}_{#1}\left (#2 \right )}
\newcommand{\LT}[2]{\operatorname{LT}_{\operatorname{#1}}\left (#2 \right )}
\newcommand{\NF}[2]{\operatorname{NF}_{\operatorname{#1}}\left (#2 \right )}

\newcommand{\stair}[1]{E \left ( #1 \right )}

\title{Polynomial Systems Solving by Fast Linear Algebra.}

\author{Jean-Charles Faugère\footnotemark[2] \and Pierrick Gaudry\footnotemark[3] \and Louise Huot\footnotemark[2] \and Guénaël Renault\footnotemark[2]} 

\date{}

\begin{document}

\maketitle

\footnotetext[2]{UPMC, Universit\'e Paris 06; INRIA, Paris Rocquencourt Center; PolSys Project, LIP6/CNRS; UMR 7606, France; Email addresses: Jean-Charles.Faugere@inria.fr, \{Louise.Huot,Guenael.Renault\}@lip6.fr}
\footnotetext[3]{Université de Lorraine; LORIA, Lorraine; CARAMEL Project, LORIA/CNRS; UMR 7503, France; Email address: Pierrick.Gaudry@loria.fr}


\begin{abstract}
Polynomial system solving is a classical problem in mathematics with a
wide range of applications. This makes its complexity a fundamental
problem in computer science.  Depending on the context, solving has
different meanings.  In order to stick to the most general case, we
consider a representation of the solutions from which one can easily
recover the exact solutions or a certified approximation of
them. Under generic assumption, such a representation is given by the
lexicographical Gröbner basis of the system and consists of a set of
univariate polynomials. The best known algorithm for computing the
lexicographical Gröbner basis is in $\widetilde{O}(d^{3n})$ arithmetic
operations where $n$ is the number of variables and $d$ is the maximal
degree of the equations in the input system. The notation
$\widetilde{O}$ means that we neglect polynomial factors in $n$. We
show that this complexity can be decreased to $\widetilde{O}(d^{\omega
  n})$ where $2 \leq \omega < 2.3727$ is the exponent in the
complexity of multiplying two dense matrices. Consequently, when the
input polynomial system is either generic or reaches the Bézout bound,
the complexity of solving a polynomial system is decreased from
$\widetilde{O}(D^3)$ to $\widetilde{O}(D^\omega)$ where $D$ is the
number of solutions of the system.  To achieve this result we propose
new algorithms which rely on fast linear algebra. When the degree of
the equations are bounded uniformly by a constant we propose a
deterministic algorithm. In the unbounded case we present a Las Vegas
algorithm.
\end{abstract}

\section{Introduction}


\paragraph{Context.}
Polynomial systems solving is a classical problem in mathematics.  It
is not only an important problem on its own, but it also has a wide
spectrum of applications.
It spans several research disciplines such as coding
theory \cite{LoYo97,BoPe99}, cryptography \cite{BPW06,Jou13},
computational game theory \cite{Dat03,Stu02},
optimization \cite{GrSa11}, \textit{etc}.
The ubiquitous nature of the problem positions the study of its complexity
at the center of theoretical computer science.
\textit{Exempli gratia}, 
in the context of computational geometry, a step of the algorithm by
Safey el Din and Schost \cite{roadmap2}, the first algorithm with better
complexity than the one by Canny \cite{canny} for solving the road map
problem, depends on solving efficiently polynomial systems.  In
cryptography, the recent breakthrough algorithm due to
Joux \cite{Jou13} for solving the discrete logarithm problem in finite
fields of small characteristic heavily relies on the same capacity.
However, depending on the context, {\em solving a polynomial system}
has different meanings.  If we are working over a finite field, then
{\em solving} generally means that we enumerate all the possible
solutions lying in this field.  On the other hand, if the field is of
characteristic zero, then {\em solving} might mean that we approximate
the real (complex) solutions up to a specified precision.
Therefore, an algorithm for solving polynomial systems should provide
an output that is valid in all contexts.  In this paper we present an
efficient algorithm to tackle the PoSSo (\textit{Polynomial Systems
Solving}) problem, the ouput of which is a representation of the roots
suitable in all the cases. The precise definition of the problem is as
follows:
\begin{problem}[PoSSo]\label{pb:posso}
Let $\K$ be the rational field $\Q$ or a finite field $\F_q$. Given a
set of polynomial equations with a finite number of solutions which
are all simple
$$
\mathcal{S} : \{f_1 = \cdots = f_s = 0\}
$$
with $f_1,\ldots,f_s \in \K[x_1,\ldots,x_n]$, find a univariate
polynomial representation of the solutions of
$\mathcal{S}$ \textit{i.e.}  $h_1,\ldots,h_n \in \K[x_n]$ such that
the system $\{ x_1 - h_1 = \cdots = x_{n-1} - h_{n-1} = h_n = 0\}$
have the same solutions as $\mathcal{S}$.
\end{problem}

It is worth noting that enumerating the solutions in a finite field or
approximating the solutions in the characteristic zero case can be
easily done once the underlying PoSSo problem is solved. Actually,
from a given univariate polynomial representation $\{ x_1 - h_1
= \cdots = x_{n-1} - h_{n-1} = h_n = 0\}$ one just have to find the
(approximated) roots of the univariate polynomial $h_n$. The
algorithms to compute such roots have their complexities in function
of $D$, the degree of $h_n$, well handled and in general they are
negligible in comparison to the cost of solving the PoSSo
problem. Note that $D$ is also the total number of solutions of the
polynomial system. For instance, if $\K = \F_q$ is a finite field, the
enumeration of the roots lying in $\F_q$ of $h_n$ can be done in
$\widetilde{O}(D)$ arithmetic operations where the notation
$\widetilde{O}$ means that we neglect logarithmic factors in $q$ and
$D$, see \cite{GaGe03}. In the characteristic zero case, finding an
approximation of all the real roots of $h_n$ can also be done in
$\widetilde{O}(D)$ where, in this case, we neglect logarithmic
factors in $D$, see \cite{Pan02}.

A key contribution to the PoSSo problem is the multivariate resultant
introduced by Macaulay in the beginning of the 20th
century \cite{Mac16}. The next major achievement on PoSSo appeared in
the 1960s when Buchberger introduced, in his PhD thesis, the concept
of Gröbner bases and the first algorithm to compute them.  Since then,
Gröbner bases have been extensively studied (see for
instance \cite{BaSt87,CLO07,Stu02,LaLa91}) and have become a powerful
and a widely used tool to solve polynomial systems. A major complexity
result related to the PoSSo problem has been shown by Lakshman and
Lazard in \cite{LaLa91} and states that this problem can be solved in
a simply exponential time in the maximal degree $d$ of the
equations \textit{i.e.} in $O(d^{O(n)})$ arithmetic operations where
$n$ is the number of variables.  As the number of solutions can be
bounded by an exponential in this degree thanks to the Bézout bound,
this result yields the first step toward a polynomial complexity in
the number of solutions for the PoSSo problem. In our context, the
Bézout bound can be stated as follows.

\begin{bezout}
Let $f_1,\ldots,f_s \subset \K[x_1,\ldots,x_n]$ and let
$d_1,\ldots,d_s$ be their respective degree. The PoSSo problem has at
most $\prod_{i=1}^s d_i$ solutions in an algebraic closure of $\K$
and counted with multiplicities.
\end{bezout}

The Bézout bound is \emph{generically} reached \textit{i.e.} $D
= \prod_{i=1}^s d_i$. We mean by generically that the system is
generic that is to say, given by a sequence of dense polynomials whose
coefficients are unknowns or any random instantiations of these
coefficients.

Whereas for the particular case of approximating or computing a
rational parametrization of all the solutions of a polynomial system
with coefficients in a field of characteristic zero there exist
algorithms with sub-cubic complexity in $D$ (if the number of real
roots is logarithmic in $D$ then $\widetilde{O}(12^nD^2)$ for the
approximation, see \cite{MoPa98}, and if the multiplicative structure
of the quotient ring is known $O \left (n2^nD^{\frac{5}{2}}\right )$
for the rational parametrization, see \cite{BSS03}). To the best of
our knowledge, there is no better bound than $O(nD^3)$ for the
complexity of computing a univariate polynomial representation of the
solutions. According to the Bézout bound the optimal complexity to
solve the PoSSo problem is then polynomial in the number of
solutions. One might ask whether the existence of an algorithm with
(quasi) linear complexity is possible.  Consider the simplest case of
systems of two equations $\{f_1 = f_2 = 0\}$ in two variables. Solving
such a system can be done by computing the resultant of the two
polynomials with respect to one of the variables.
 From \cite{GaGe03}, the complexity of computing such a resultant is
polynomial in the Bézout bound with exponent strictly greater than
one. In the general case \textit{i.e.} more than two variables, the
PoSSo problem is much more complicated. Consequently, nothing
currently suggests that a (quasi) linear complexity is possible.  The
main goal of this paper is to provide the first algorithm with
sub-cubic complexity in $D$ to solve the PoSSo problem, which is
already a noteworthy progress. More precisely, we show that when the
Bézout bound is reached, the complexity to solve the PoSSo problem is
polynomial in the number of solutions with exponent $2 \leq \omega <
3$, where $\omega$ is the exponent in the complexity of multiplying
two dense matrices. Since the 1970s, a fundamental issue of
theoretical computer science is to obtain an upper bound for $\omega$
as close as possible to two. In particular, Vassilevska Williams
showed in 2011 \cite{Vas12} that $\omega$ is upper bounded by
$2.3727$ \textit{i.e.} $2 \leq \omega < 2.3727$. By consequence, our
work tends to show that a quadratic complexity in the number of
solutions for the PoSSo problem can be expected. A direct consequence
of such a result is the improvement of the complexity of many
algorithms requiring to solve the PoSSo problem, for instance in
asymmetric \cite{Gau09,FGHR13} or symmetric \cite{BPW06,BPW06b}
cryptography.

\paragraph{Related works.}
In order to reach this goal we develop new algorithms in Gr\"obner
basis theory. Let $\mathcal{S}$ be a polynomial system in
$\K[x_1,\dots,x_n]$ verifying the hypothesis of
Problem~\ref{pb:posso}, \textit{i.e.} with a finite number of
solutions in an algebraic closure of $\K$ which are all simple. A Gröbner basis is to $\mathcal{S}$ what
row echelon form is to a linear system. For a fixed monomial ordering,
given a system of polynomial equations, its associated Gröbner basis
is unique after normalization. From an algorithmic point of view,
monomial orderings may differ: some are attractive for the efficiency
whereas some others give rise to a more structured output. Hence, the fastest monomial ordering is usually the
degree reverse lexicographical ordering, denoted DRL. However, in
general, a DRL Gröbner basis does not allow to list the solutions of
\(\mathcal{S}\). An
important ordering which provides useful outputs is the
lexicographical monomial ordering, denoted LEX in the sequel.
Actually, for a characteristic $0$ field or with a sufficiently large
one, up to a linear change of the coordinates, a Gr\"obner basis for
the LEX ordering of the polynomial system $\mathcal{S}$ gives a
univariate polynomial representation of its solutions \cite{GiMo89,
Lak90}. That is to say, computing this Gr\"obner basis is equivalent
to solving the PoSSo problem~\ref{pb:posso}. It is usual to define the
following: the ideal generated by $\mathcal{S}$ is said to be
in \emph{Shape Position} when its LEX Gr\"obner basis if of the form
$\{x_1 - h_1(x_n),\ldots,x_{n-1} - h_{n-1}(x_n), h_n(x_n)\}$ where
$h_1,\ldots,h_{n-1}$ are univariate polynomials of degree less than
$D$ and $h_n$ is a univariate polynomial of degree $D$ (\textit{i.e.}
one does not need to apply any linear change of coordinates to get the
univariate polynomial representation). In the first part of this paper,
we will avoid the consideration of the probabilistic choice of the
linear change of coordinates in order to be in Shape Position, thus we
assume the following hypothesis.

\begin{hypothesis}\label{hyp:shapeposition}
Let $\mathcal{S} \subset \K[x_1,\ldots,x_n]$ be a polynomial system
with a finite number of solutions which are all simple. Its associated
LEX Gröbner basis is in \textit{Shape Position}.
\end{hypothesis}

From a DRL Gröbner basis,
one can compute the corresponding LEX Gröbner basis by using a change
of ordering algorithm. Consequently, when the associated LEX Gröbner
basis of the system $\mathcal{S}$ is in \textit{Shape
Position} \textit{i.e.} $\mathcal{S}$ satisfies
Hypothesis~\ref{hyp:shapeposition} the usual and most efficient algorithm
is first to compute a DRL Gröbner basis. Then,  the LEX Gröbner basis is computed by using a
change of ordering algorithm. This is summarized in
Algorithm~\ref{usual_posso}.

\begin{algorithm}[!ht]
\SetKwData{Left}{left}\SetKwData{This}{this}\SetKwData{Up}{up}
\SetKwFunction{Union}{Union}\SetKwFunction{FindCompress}{FindCompress}
\SetKwInOut{Input}{Input}\SetKwInOut{Output}{Output}
  \Input{A polynomial system $\mathcal{S} \subset \K[x_1,\ldots,x_n]$ which satisfies Hypothesis~\ref{hyp:shapeposition}.}
  \Output{The LEX Gröbner basis of $\mathcal{S}$ \textit{i.e.} the univariate polynomial representation of the solutions of $\mathcal{S}$. }
  \caption{Solving polynomial systems}\label{usual_posso}
  Computing the DRL Gröbner basis of $\langle S \rangle$\;
  From the DRL Gröbner basis, computing the LEX Gröbner basis of $\langle S \rangle$\;
  \Return The LEX Gröbner basis of $\mathcal{S}$\;
\end{algorithm}

The first step of Algorithm~\ref{usual_posso} can be done by using
$F_4$ \cite{Fau99} or $F_5$ \cite{Fau02} algorithms.  The complexity
of these algorithms for \emph{regular systems} is well handled. For
the homogeneous case, the regular property for a polynomial system
$\{f_1,\ldots,f_s\} \subset \K[x_1,\ldots,x_n]$ is a generic property
which implies that for all $i \in \{2,\ldots,s\}$, the polynomial
$f_i$ does not divide zero in the quotient ring
$\K[x_1,\ldots,x_n]/\langle f_1,\ldots,f_{i-1} \rangle$. There is an
analogous definition for the affine case, see
Definition~\ref{def:regular}. For the particular case of the DRL
ordering, computing a DRL Gröbner basis of a regular system in
$\K[x_1,\ldots,x_n]$ with equations of same degree, $d$, can be done
in $\widetilde{O}(d^{\omega n})$ arithmetic operations
(see \cite{BFSY05,Laz83}). Moreover, the number of solutions $D$ of
the system can be bounded by $d^n$ by using the Bézout bound. Since,
this bound is generically (\textit{i.e.} almost always)
reached \textit{i.e.} $D = d^n$, computing a DRL Gröbner basis can be
done in $\widetilde{O}(D^\omega)$ arithmetic operations.  Hence, in
this case the first step of Algorithm~\ref{usual_posso} has a
polynomial arithmetic complexity in the number of solutions with
exponent $\omega$.

The second step of Algorithm~\ref{usual_posso} can be done by using a
change of ordering algorithm.  In 1993, Faugère \textit{et al.} showed
in \cite{FGLM93} that change of ordering for zero dimensional ideals
is closely related to linear algebra.  Indeed, they proposed a change
of ordering algorithm, denoted FGLM in the literature, which proceeds
in two stages. Let $\G_{>_1}$ be the given Gröbner basis w.r.t. the
order $>_1$ of an ideal in $\K[x_1,\ldots,x_n]$. First, we need for
each $i \in \{1,\ldots,n\}$ a matrix representation, $T_i$, of the
linear map of
$\K[x_1,\ldots,x_n]/\Id{\G_{>_1}}\rightarrow\K[x_1,\ldots,x_n]/\langle
\G_{>_1} \rangle$ corresponding to the multiplication by $x_i$. The matrix $T_i$ is called multiplication matrix by \(x_{i}\).
These matrices are constructed by computing $O(nD)$ matrix-vector
products (of size $D\times D\text{ times }D\times 1$). Hence, the
first stage of FGLM algorithm (Algorithm~\ref{fglm}) has an arithmetic
complexity bounded by $O(nD^3)$. Once all the multiplication matrices
are computed, the second Gröbner basis w.r.t. the new monomial order
$>_2$ is recovered by testing linear dependency of $O(nD)$ vectors of
size \(D\times 1\). This can be done in $O(nD^3)$ arithmetic
operations. This algorithm is summarized in
Algorithm~\ref{fglm}. Therefore, in the context of the existing
knowledge, solving regular zero-dimensional systems can be done in
$O(nD^3)$ arithmetic operations and change of ordering appears as the
bottleneck of PoSSo.

\begin{algorithm}[!ht]
\SetKwData{Left}{left}\SetKwData{This}{this}\SetKwData{Up}{up}
\SetKwFunction{Union}{Union}\SetKwFunction{FindCompress}{FindCompress}
\SetKwInOut{Input}{Input}\SetKwInOut{Output}{Output}
  \Input{The Gröbner basis w.r.t. $>_1$ of an ideal $\id{I}$.}
  \Output{The Gröbner basis w.r.t. $>_2$ of $\id{I}$.} 
  \caption{FGLM}\label{fglm}
  Computing the multiplication matrices $T_1,\ldots,T_n$;\hspace*{2.75cm}
  {\scriptsize // \textit{$O(nD)$ matrix-vector products}}\\\nllabel{step1fglm}
  From $T_1,\ldots,T_n$ computing the Gröbner basis of $\id{I}$ w.r.t. $>_2$;
  \hspace*{1.1cm}{\scriptsize // \textit{$O(nD)$ linear dependency tests}}\nllabel{step2fglm}
\end{algorithm}

\paragraph{Fast Linear Algebra.}
Since the second half of the 20th century, an elementary issue in
theoretical computer science was to decide if most of linear algebra
problems can be solved by using fast matrix multiplication and
consequently bound their complexities by that of multiplying two dense
matrices \textit{i.e.} $O(m^\omega)$ arithmetic operations where
\(m\times m\) is the size of the matrix and $2 \leq \omega < 2.3727$. For instance, Bunch and Hopcroft showed
in \cite{BuHo74} that the inverse or the triangular decomposition can
be done by using fast matrix multiplication. Baur and Strassen
investigated the determinant in \cite{BaSt83}. The case of the
characteristic polynomial was treated by Keller-Gehrig
in \cite{Kel85}. Although that the link between linear algebra and the
change of ordering has been highlighted for several years, relating
the complexity of the change of ordering with fast matrix
multiplication complexity is still an open issue. \paragraph{Main
results.}The aim of this paper is then to give an initial answer to
this question in the context of polynomial systems
solving \textit{i.e.}  for the special case of the DRL and LEX
orderings. More precisely, our main results are summarized in the
following theorems.  First we present a \emph{deterministic} algorithm
computing the univariate polynomial representation of a polynomial
system verifying Hypothesis~\ref{hyp:shapeposition} and whose
equations have bounded degree.

\begin{theorem}\label{thm:determinist}
Let $\mathcal{S} = \{ f_1,\ldots,f_n \} \subset \K[x_1,\ldots,x_n]$ be
a polynomial system verifying Hypothesis~\ref{hyp:shapeposition} and
let $\K$ be the rational field $\Q$ or a finite field $\F_q$. If the
sequence $(f_1,\ldots,f_n)$ is a regular sequence and if the degree of
each polynomial $f_i$ ($i=1,\ldots,n$) is uniformly bounded by a fixed
integer $d$ then there exists a deterministic algorithm solving
Problem~\ref{pb:posso} in $\widetilde{O}(d^{\omega n} + D^\omega)$
arithmetic operations where the notation $\widetilde{O}$ means that we
neglect logarithmic factors in $D$ and polynomial factors in $n$ and
$d$.
\end{theorem}

Then we present a \emph{Las Vegas} algorithm extending the result of
Theorem~\ref{thm:determinist} to polynomial systems not necessarily
verifying Hypothesis~\ref{hyp:shapeposition} and whose equations have
non fixed degree.

\begin{theorem}\label{thm:lasvegas}
Let $\mathcal{S} = \{ f_1,\ldots,f_n \} \subset \K[x_1,\ldots,x_n]$ be
a polynomial system and let $\K$ be the rational field $\Q$ or a
finite field $\F_q$. If the sequence $(f_1,\ldots,f_n)$ is a regular
sequence where the degree of each polynomial is uniformly bounded by a
non fixed parameter $d$ then there exists a Las Vegas algorithm solving
Problem~\ref{pb:posso} in $\widetilde{O}(d^{\omega n} + D^\omega)$
arithmetic operations; where the notations $\widetilde{O}$ means that
we neglect logarithmic factors in $D$ and polynomial factors in $n$.
\end{theorem}

If $\K = \Q$ the probability of failure of the algorithm mentioned in
Theorem~\ref{thm:lasvegas} is zero while in the case of a finite field
$\F_q$ of characteristic $p$, it depends on the size of $p$ and $q$,
see Section~\ref{newposso}.

As previously mentioned, the Bézout bound allows to bound $D$ by
$d^{n}$ and generically this bound is reached \text{i.e.} $D =
d^{n}$. By consequence, Theorem~\ref{thm:determinist} (respectively
Theorem~\ref{thm:lasvegas}) means that if the equations have fixed
(respectively non fixed) degree then there exists a deterministic
(respectively a Las Vegas) algorithm computing the univariate
polynomial representation of generic polynomial systems in
$\widetilde{O}(D^\omega)$ arithmetic operations.

To the best of our knowledge, these complexities are the best ones for
solving the PoSSo Problem~\ref{pb:posso}. For example, in the case of
field of characteristic zero, under the same hypotheses as in
Theorem~\ref{thm:determinist}, one can now compute a univariate
polynomial representation of the solutions in
$\widetilde{O}(D^\omega)$ without assuming that the multiplicative
structure of $\K[x_1,\ldots,x_n]$ is known. This can be compared to
the method in \cite{BSS03} which, assuming the multiplicative
structure of the quotient ring known, computes a parametrization of
the solutions in $O \left (n2^nD^\frac{5}{2} \right )$. Noticing that
under the hypotheses of Theorem~\ref{thm:determinist}, $n$ is of the
order of $\log_2(D)$ and the algorithm in \cite{BSS03} has a
complexity in $\widetilde{O}\left (D^{\frac{7}{2}}\right )$.

\paragraph{Importance of the hypotheses.}
The only two hypotheses which limits the applicability of the
algorithms in a meaningful way is that (up to a linear change of
variables) the ideal admits a LEX Gröbner basis in \emph{Shape
Position} and that the number of solutions in an algebraic closure of
the coefficient field counted with multiplicity is finite. The other
hypotheses are stated either to simplify the paper or to simplify the
complexity analysis. More precisely, the hypothesis that the solutions
are all simple is minor. Indeed, it is sufficient to get the required
hypothesis about the shape of the LEX Gröbner basis but not
necessary. The hypothesis stating the regularity of the system is
required to get a complexity bound on the computation of the first
(DRL) Gröbner basis. Indeed, without this hypothesis the computation
of the first Gröbner basis is possible but there is no known
complexity analysis of such a computation. It is a common assumption
in algorithmic commutative algebra. The assumption on the degree of
the equations in the input system is stated in order to obtain a
simply form of the complexity of computing the first Gröbner
basis \textit{i.e.}  $\widetilde{O}(d^{\omega n})$. Finally, the
hypothesis of genericity (\textit{i.e.} the Bézout bound is reached) is
required to express the complexity of the computation of the first
Gröbner basis in terms of the number of solutions \textit{i.e.}
$\widetilde{O}(d^{\omega n}) = \widetilde{O}(D^\omega)$.  We would
like to precise that all the complexities in the paper are given in
the worst case for all inputs with the required assumptions.

\paragraph{Outline of the algorithms.}
In 2011, Faugère and Mou proposed in \cite{FaMo11} another kind of
change of ordering algorithm to take advantage of the sparsity of the
multiplication matrices. Nevertheless, when the multiplication
matrices are not sparse, the complexity is still in $O(D^3)$
arithmetic operations. Moreover, these complexities are given assuming
that the multiplication matrices have already been computed and the
authors of \cite{FaMo11} do not investigate their computation whose
complexity is still in $O(nD^3)$ arithmetic operations. In FGLM, the
matrix-vectors products (respectively linear dependency tests) are
intrinsically sequential.  This dependency implies a sequential order
for the computation of the matrix-vectors products (respectively
linear dependency tests) on which the correctness of this algorithm
strongly relies.  Thus, in order to decrease the complexity to
$\widetilde{O}\left (D^\omega \right )$ we need to propose new
algorithms.

To achieve result in Theorem~\ref{thm:determinist} we propose two
algorithms in \(\widetilde{O}(D^\omega)\), each of them corresponding
to a step of the Algorithm~\ref{fglm}.

We first present an algorithm to compute multiplication matrices
assuming that we have already computed a Gr\"obner basis \(\G\).  The
bottleneck of the existing algorithm~\cite{FGLM93} came from the fact
that \(n D\) normal forms have to be computed in a sequential order.
The key idea is to show that we can compute \emph{simultaneously} the
normal form of all monomials \emph{of the same degree} by computing
the row echelon form of a well chosen matrix.  Hence, we replace
the \(n D\) normal form computations by \(\log_2(D)\) (we iterate
degree by degree) row echelon forms on matrices of size \((n\ D)\times
(n D+D)\). To compute simultaneously these normal forms we observe
that if $r$ is the normal form of a monomial $m$ of degree \(d-1\)
then $m - r$ is a polynomial in the ideal of length at most \(D+1\);
then we generate the Macaulay matrix of all the products \(x_{i
}m-x_{i}r\) (for \(i\) from \(1\) to \(n\)) together with the
polynomials \(g\) in the Gr\"obner basis \(\G\) of degree
exactly \(d\). We recall that the Macaulay matrix of some
polynomials \cite{Laz83,Mac16} is a matrix whose rows consist of the
coefficients of these polynomials and whose columns are indexed with
respect to the monomial ordering. Computing a row echelon form of the
concatenation of all the Macaulay matrices in degree less or equal to
$d$ enable us to obtain all the normal forms of all monomials of
degree \(d\). This yields an algorithm to compute the multiplication
matrices of arithmetic complexity $O(\delta n^\omega D^\omega)$ where
$\delta$ is the maximal degree of the polynomials in \(\G\); note that
this algorithm can be seen as a redundant version of \(F_{4}\)
or \(F_{5}\).

In order to prove Theorem~\ref{thm:lasvegas} we use the fact that, in
a generic case, only the multiplication matrix by the \emph{smallest
variable} is needed. Surprisingly, we show (Theorem~\ref{thm:Tnopt1})
that, in this generic case,\emph{ no arithmetic} operation is required
to build the corresponding matrix. Moreover, for non generic
polynomial systems, we prove (Corollary~\ref{cor:Tnopt2}) that a
generic linear change of variables bring us back to this
case.

The second algorithm (step 2 of Algorithm~\ref{fglm}) we describe is
an adaptation of the algorithm given in~\cite{FaMo11} when the ideal
is in \textit{Shape Position}. Once again only the multiplication
matrix by the \emph{smallest variable} is needed in this case. When
the multiplication matrix \(T\) of size \(D\times D\) is dense,
the \(O(D^3)\) arithmetic complexity in~\cite{FaMo11} came from
the \(2D\) matrix-vector products $T^i \mathbf{r}$ for $i=1,\ldots,2D$
where $\mathbf{r}$ is a column vector of size $D$. To decrease the
complexity we follow the Keller-Gehrig algorithm~\cite{Kel85}: first,
we compute \(T^2, T^4, \ldots,$ $T^{2^{\lceil \log_2 D \rceil}}\)
using binary powering; second, all the products \(T^{i}\mathbf{r}\)
are recovered by computing \(\log_2D\) matrix multiplications. Then,
in the Shape Position case, the \(n\) univariate polynomials of the
lexicographical Gr\"obner basis are computed by solving \(n\)
structured linear systems (Hankel matrices) in \(O(n D\log_2^2(D))\)
operations. We thus obtain a change of ordering algorithm (DRL to LEX
order) for \textit{Shape Position} ideals whose complexity is in
$O \left (\log_2(D)\left (D^\omega + n\log_2(D)D \right )\right )$
arithmetic operations.

\paragraph{Organization of the paper.} 
The paper is organized as follows.  In Section~\ref{sec:notations} we
first introduce some required notations and backgrounds. Then, an
algorithm to compute the LEX Gröbner basis given the multiplication
matrices is presented in Section~\ref{sec:change}. Next, we describe
the algorithm to compute multiplication matrices in
Section~\ref{sec:matrix}. Afterwards, their complexity analysis are
studied in Section~\ref{sec:tame} where we obtain
Theorem~\ref{thm:determinist}. Finally, in Section~\ref{sec:wild}
we show how to deduce (\textit{i.e.} without any costly arithmetic  
operation) the multiplication matrix by the smallest variable.
According to this construction we propose another algorithm for
polynomial systems solving which allows to obtain the result in
Theorem~\ref{thm:lasvegas}. In Appendix~\ref{appendix:bench} we
discuss about the impact of our algorithm on the practical solving of
the PoSSo problem.

The
authors would like to mention that a preliminary version of this work
was published as a poster in the ISSAC 2012 conference \cite{FGHR12}.

\section{Notations and preliminaries}\label{sec:notations}


Throughout this paper, we will use the following notations.  Let $\K$
denote a field (for instance the rational numbers $\Q$ or a finite
field $\F_q$ of characteristic $p$), and $\A =
\K[x_1,\ldots,x_n]$ be the polynomial ring in $n$ variables with $x_1 > \cdots >
x_n$.
A monomial of $\K[x_1,\ldots,x_n]$ is a
product of powers of variables and a term is a product of a monomial
and a coefficient in $\K$. We denote by $\text{LT}_<(f)$ the leading
term of $f$ w.r.t. the monomial ordering $<$.

 Let $\id{I}$ be an ideal of $\A$; once a monomial ordering $<$
is fixed, a reduced Gröbner basis $\G_{<}$ of $\id{I}$ w.r.t. $<$ can
be computed. 

\begin{definition}[Gröbner basis]
Given a monomial ordering $<$ and an ideal $\id{I}$ of $\A$, a finite
subset $\G_{<} = \{g_1,\ldots,g_s\}$ of $\id{I}$ is a Gröbner basis of
$\id{I}$ w.r.t. the monomial ordering $<$ if the ideal
$\{\text{LT}_<(f)\ |\ f
\in \id{I}\}$ is generated by
$\{\text{LT}_<(g_1),\ldots,\text{LT}_<(g_s)\}$. The Gröbner basis
$\G_{<}$ is the unique reduced Gröbner basis of $\id{I}$ w.r.t. the
monomial ordering $<$ if $g_1,\ldots,g_s$ are monic polynomials and
for any $g_i \in \G_{<}$ all the terms in $g_i$ are not divisible by a
leading term of $g_j$ for all $g_j \in \G_{<}$ such that $j \ne i$.
\end{definition}

We always consider reduced Gröbner basis so henceforth, we omit the
adjective ``reduced''.  For instance, $\gdrl$ (resp. $\glex$)\ denotes
the Gröbner basis of $\id{I}$ w.r.t. the DRL order (resp. the LEX
order). In particular, a Gröbner basis $\G_{<} = \{g_1,\ldots,g_s\}$
of an ideal $\id{I} = \Id{f_1,\ldots,f_m}$ is a basis of
$\id{I}$. Hence, solving the system $\{g_1,\ldots,g_s\}$ is equivalent
to solve the system $\{f_1,\ldots,f_m\}$.

\begin{definition}[Zero-dimensional ideal]
Let $\id{I}$ be an ideal of $\A$. If $\id{I}$ has a finite number of
solutions, counted with multiplicities in an algebraic closure of
$\K$, then $\id{I}$ is said to be zero-dimensional. This number,
denoted by \(D\), is also the degree of the ideal $\id{I}$. If
$\id{I}$ is zero-dimensional, then the residue class ring $V_{\id{I}}
= \A / \id{I}$ is a $\K$-vector space of dimension $D$.
\end{definition}

From $\G_{<}$ one can deduced a vector basis of $V_{\id{I}}$.  Indeed,
the canonical vector basis of $V_{\id{I}}$ is $B = \{ 1 = \epsilon_1
< \cdots < \epsilon_D \}$ where \(\epsilon_i\) are irreducible
monomials (that is to say for all $i \in \{1,\ldots,D\}$, there is no
$g \in \G_{<}$ such that $\text{LT}_<(g)$ divides $\epsilon_i$).

\begin{definition}[Normal Form]
Let $f$ be a polynomial in $\A$. The normal form of $f$ is defined
w.r.t. a monomial ordering $<$ and
denoted \(\text{NF}_{<}(f)\): \(\text{NF}_{<}(f)\) is the unique
polynomial in $\A$ such that no term of $\text{NF}_{<}(f)$ is
divisible by a leading term of a polynomial in $\G_{<}$ and there
exists $g \in \id{I}$ such that $f = g + \text{NF}_{<}(f)$. That is to
say, $\text{NF}_{<}$ is a (linear) projection of $\A$ on $V_{\id{I}}$.
We recall that for any polynomials \(f,g,h\) we have \(\text{NF}_{<}(f
g)=\text{NF}_{<}(\text{NF}_{<}(f)
g)=\text{NF}_{<}(\text{NF}_{<}(f) \text{NF}_{<}(g))\).
\end{definition}

Let $\psi$ be the representation of $V_{\id{I}}$ as a subspace of
$\K^D$ associated to the canonical basis \(B\):
\[
\psi :\left(\begin{array}{cccc}
 & V_{\id{I}} & \rightarrow & \K^D\\
& \sum_{i=1}^D \alpha_i \epsilon_i & \mapsto & [\alpha_1,\ldots,\alpha_D]^{t}\,.
\end{array}\right)
\]
We call \emph{multiplication matrices}, denoted $T_1,\ldots,T_n$, the matrix
representation of the multiplication by $x_1,\ldots,x_n$ in
$V_{\id{I}}$.  That is to say, the $i^{\text{th}}$ column of the
matrix $T_j$ is given by $\psi(\text{NF}_{<}{(\epsilon_i x_j)}) =
[c_{i,1}^{(j)},\ldots,c_{i,D}^{(j)}]^{t}$ hence, $T_k = \left
(c_{i,j}^{(k)} \right )_{i,j=1,\ldots,D}$.

The LEX Gröbner basis of an ideal $\id{I}$ has a
triangular form. In particular, when $\id{I}$ is zero-dimensional, its
LEX Gröbner basis always contains a univariate
polynomial. In general, the expected form of a LEX Gröbner
basis is the \textit{Shape Position}. 
When the field $\K$ is $\Q$ or when its characteristic $p$ is
sufficiently large, almost all zero-dimensional ideals have, up to a
linear change of coordinates, a LEX Gröbner basis in \textit{Shape
Position} \cite{KobFuFu}. A characterization of the zero-dimensional
ideals that can be placed in shape position has been given
in \cite{shapepos}. A less general result \cite{GiMo89,Lak90} usually
called the \emph{Shape Lemma} is the following: an ideal \(\id{I}\) is
said to be radical if for any polynomial
in \(\A\), \(f^{k}\in \id{I}\) implies \(f\in \id{I}\). Up to a linear
change of coordinates, any radical ideal has a LEX Gröbner basis
in \textit{Shape Position}. From now on, all the ideals considered in
this paper will be zero-dimensional and will have a LEX Gröbner basis
in \textit{Shape Position}. Moreover, we fix the DRL order for the
basis of $V_{\id{I}}$ that is to say that $B
= \{\epsilon_1,\ldots,\epsilon_D\}$ will always denote the canonical
vector basis of $V_{\id{I}}$ w.r.t. the DRL order. Since
for \textit{Shape Position} ideals the LEX Gröbner basis is described
by $n$ univariate polynomials we will call it the ``univariate
polynomial representation'' of the ideal or, up to multiplicities, of
its variety of solutions.

In the following section, we present an algorithm to compute the LEX
Gröbner basis of a \textit{Shape Position} ideal. This algorithm
assumes the DRL Gröbner basis and a multiplication matrix to be
known. The computation of the multiplication matrices is treated in
Section~\ref{sec:matrix}.

\section{Univariate polynomial representation using structured linear algebra}\label{sec:change}


In this section, we present an algorithm to compute univariate
polynomial representation. This algorithm follows the one described
in \cite{FaMo11}. The main difference is that this new algorithm and
its complexity study do not take into account any structure of the
multiplication matrices (in particular any sparsity assumption).

Let $\glex = \{ h_n(x_n),x_{n-1} - h_{n-1}(x_n),\ldots,x_{1} -
h_{1}(x_n) \}$ be the LEX Gröbner basis of $\id{I}$. Given the
multiplication matrices $T_1,\ldots,T_n$, an algorithm to compute the
univariate polynomial representation has to find the $n$ univariate
polynomials $h_1,\ldots,h_n$.  For this purpose, we can proceed in two
steps. First, we will compute $h_n$. Then, by using linear algebra
techniques, we will compute the other univariate polynomials
$h_1,\ldots,h_{n-1}$.

\begin{remark}\label{rem:deterministic}
In this section, for simplicity, we present a probabilistic algorithm
to compute the univariate polynomial representation. However, to
obtain a deterministic algorithm it is sufficient to adapt the
deterministic algorithm for radical ideals admitting a LEX Gröbner
basis in \textit{Shape Position} given in \cite{SparseFGLM} in exactly the
same way we adapt the probabilistic version.
\end{remark}

\subsection{Computation of $h_n$}

To compute $h_n$ we have to compute the minimal polynomial of $T_n$.
To this end, we use the first part of the  Wiedemann probabilistic algorithm which succeeds
with good probability if the field $\K$ is sufficiently large,
see \cite{Wie86}.

Let $\mathbf{r}$ be a random column vector in $\K^D$ and $\mathbf{1}
= \psi(1)^t=[1,0,\ldots,0]^t$. If $a = [a_1,\ldots,a_D]$ and $b =
[b_1,\ldots,b_D]$ are two vectors of $\K^D$, we denote by $(a,b)$ the
dot product of $a$ and $b$ defined by $(a,b) =
\sum_{i=1}^D a_ib_i$.
If \(\mathbf{r}_{1},\ldots,\mathbf{r}_{k}\) are column vectors then we denote by \((\mathbf{r}_{1}| \ldots |\mathbf{r}_{k})\) the matrix with \(D\) rows and \(k\) columns  obtained by joining the vectors \(\mathbf{r}_i\) vertically.

Let $S = [ (\mathbf{r}, T_n^j \mathbf{1})\ |\ j =
0,\ldots,2D-1 ]$ be a linearly recurrent sequence of size $2D$. By
using for instance the Berlekamp-Massey algorithm \cite{Mas69}, we can
compute the minimal polynomial of $S$ denoted $\mu$. If $\deg(\mu(x_n)) = D$ then we deduce that $\mu(x_n) =
h_n(x_n) \in \glex$ since
$\mu$ is a divisor of $f_n$.

In order to compute efficiently 
$S$, we first notice that $(\mathbf{r},T_n^j \mathbf{1}) =
(T^j
\mathbf{r},\mathbf{1})$ where  $T = T_n^t$ is the transpose matrix of $T_n$.  Then, we compute  $T^2,
T^4, \ldots,$ $T^{2^{\lceil \log_2 D \rceil}}$ using binary powering
with  \(\lceil \log_2 D \rceil\) matrix
multiplications. Similarly to ~\cite{Kel85}, the vectors
$T^j\mathbf{r}$ for $j=0,\ldots,(2D-1)$ are computed by induction
in \(\log_2 D\) steps:
\begin{equation}\label{indu}
\begin{array}{r c l}
T^2 (T \mathbf{r}\ |\ \mathbf{r}) & = & (T^3 \mathbf{r}\ |\ T^2
\mathbf{r})\\ T^4 (T^3 \mathbf{r}\ |\ T^{2}
\mathbf{r}\ |\ T \mathbf{r} |\ \mathbf{r}) & = & (T^7 \mathbf{r}\ |\ T^6
\mathbf{r}\ |\ T^5 \mathbf{r}\ |\ T^4 \mathbf{r})\\ & \vdots &
\\ T^{2^{\lceil \log_2(D) \rceil}} (T^{2^{\lceil \log_2(D) \rceil}-1} \mathbf{r}\ |\ \cdots\ |\ \mathbf{r})
& = & (T^{2D - 1} \mathbf{r}\ |\ T^{2D-2}
\mathbf{r}\ |\ \cdots\ |\ T^{2^{\lceil \log_2(D)
    \rceil}}\mathbf{r})\,. \end{array}
\end{equation}

\subsection{Recovering $h_1,\ldots,h_{n-1}$}
\label{sec:other polys}
We write $h_i = \sum_{k=0}^{D-1} \alpha_{i,k} x_n^k$ for
$i=1,\ldots,n-1$ where $\alpha_i \in \K$ are unknown. We have for 
$i=1,\ldots,n-1$:
\begin{displaymath}
x_i - h_i \in \glex  \text{ is equivalent to }  0=\NF{drl}{x_i - \sum_{k=0}^{D-1} \alpha_{i,k}x_n^k} =T_i \mathbf{1} - \sum_{k=0}^{D-1} \alpha_{i,k}T_n^k\mathbf{1}\notag\,.\\
\end{displaymath}
Multiplying the last equation by \(T_{n}^{j}\) for any \(j=0,\ldots,(D-1)\) and taking the scalar product we deduce that:
\begin{equation}
0=(\mathbf{r},T_n^j(T_i \mathbf{1})) - \sum_{k=0}^{D-1}
\alpha_{i,k}(\mathbf{r},T_n^{k+j}\mathbf{1}) = 
(T^j \mathbf{r},T_i \mathbf{1}) - \sum_{k=0}^{D-1}
\alpha_{i,k}(T^{k+j}\mathbf{r},\mathbf{1})
\label{eqMat}
\end{equation}

Hence, we can recover $h_i$, for   $i=1,\ldots,n-1$ by solving $n-1$
structured linear systems:
\begin{center}
  \begin{equation}
    \label{hankelsystem}
    \begin{tikzpicture}[>=latex, baseline=(current bounding box.center)]
      \node (sys) at (0,0) {$\displaystyle{\left ( \begin{array}{c}
            (T^0 \mathbf{r},T_i\mathbf{1} )\\ 
            (T^1 \mathbf{r},T_i\mathbf{1} )\\
            \vdots\\
            (T^{D-1} \mathbf{r},T_i\mathbf{1})
          \end{array} 
          \right ) =
          \left ( \begin{array}{c c c c }
            (T^{0} \mathbf{r},  \mathbf{1} ) & ( T^1\mathbf{r},  \mathbf{1} ) & \ldots & ( T^{D-1}\mathbf{r},  \mathbf{1} )\\
            (T^1 \mathbf{r},  \mathbf{1} ) & ( T^{2}\mathbf{r},  \mathbf{1} ) & \ldots & ( T^{D}\mathbf{r},  \mathbf{1} )\\
            \vdots & \vdots & \ddots & \vdots\\
            ( T^{D-1}\mathbf{r},  \mathbf{1} ) & ( T^{D}\mathbf{r},  \mathbf{1} ) & \ldots & ( T^{2D-2} \mathbf{r},  \mathbf{1} )
          \end{array} \right )
          \left ( \begin{array}{c}
            c_{i,0}\\
            c_{i,1}\\
            \vdots\\
            c_{i,D-1}
          \end{array} \right )}$};
      \node (bi) at (-4.8,-1.55) {$\mathbf{b_i}$};
      \node (H) at (0.7,-1.55) {$\mathcal{H}$};
      \node (ci) at (5.4,-1.55) {$\mathbf{c_i}$};
      \draw[thick,decorate,decoration={brace,mirror}] (-5.95,-1.2) -- (-3.7,-1.2);
      \draw[thick,decorate,decoration={brace,mirror}] (-2.5,-1.2) -- (4,-1.2);
      \draw[thick,decorate,decoration={brace,mirror}] (4.8,-1.2) -- (5.9,-1.2);
    \end{tikzpicture}
\end{equation}
\end{center}

Note that the linear system~\eqref{hankelsystem} has a unique solution
since from \cite{JoMa89} the rank of $\mathcal{H}$ is given by the degree of
the minimal polynomial of $S$ which is exactly $D$ in our case.
The following lemma tell us that we can compute \(T_i \mathbf{1}\) without knowing \(T_i\).
\begin{lemma}\label{nfxi}
The vectors $T_i \mathbf{1}$ for $i=1,\ldots,n-1$ can be read from
$\gdrl$.
\end{lemma}

\begin{proof}\ We have to consider the two cases \(\NF{drl}{x_i} \ne x_i\) or \(\NF{drl}{x_i} = x_i\). 

First, if $\NF{drl}{x_i} \ne x_i$ then there exists $ g \in \gdrl$ such that
$\LT{drl}{g}$ divides $x_i$. This implies that $g$ is
 a \emph{linear} equation:
 \begin{equation}
  x_i + \sum_{j>i}^n  \alpha_{i,j} x_j + \alpha_{i,0} \text{ with } \alpha_{i,j} \in \K \,. \label{lin}
\end{equation}
Hence, we have $\NF{drl}{x_i} = - \sum_{j>i}^n \alpha_{i,j} x_j
- \alpha_{i,0}$
and \(T_{i}\mathbf{1}=-[\alpha_{i,0},0,\ldots,0,\alpha_{i,i+1},\ldots,\alpha_{i,n},0,\ldots]^t\).
Otherwise, \(\NF{drl}{x_i} = x_i\) so
that \(T_{i}\mathbf{1}=[0,\ldots,0,1,0,\ldots,0]^t\).
\end{proof}
       
Hence, once the vectors \(T^{j}\mathbf{r}\) have been computed
for \(j=0,\ldots,(2D-1)\), we can deduce directly the Hankel
matrix \(\mathcal{H}\) with no computation but scalar products would
seem to be needed to obtain the vectors \(\mathbf{b}_i\).  However, by
removing the linear equations from \(\gdrl\) we can deduce the
$\mathbf{b}_i$ without arithmetic operations.

\paragraph{Linear equations in \(\gdrl\).}
\ Let denote by $\mathcal{\mathbb{L}}$ the set of polynomials  in $\gdrl$ of total degree \(1\) (usually \(\mathbb{L}\) is empty). We define   $\mathcal{L}=\{j\in \{1,\ldots,n-1\}\text{ such that } \NF{drl}{x_{j}} \neq\ x_{j}\}$ and \(\mathcal{L}^c=\{1,\ldots,n-1\}\backslash\mathcal{L}\) so that  \(\{x_i \mid i \in \mathcal{L}\}=\LT{drl}{\mathbb{L}}\). In other words there is no linear form
in $\gdrl$ with leading term $x_{i}$ when \(i\in \mathcal{L}^c\).

We first solve the linear systems ~\eqref{hankelsystem}
for \(i\in \mathcal{L}^c\): we know from the proof of Lemma~\ref{nfxi}
that $T_{i}\mathbf{1} = [0,\ldots,0,1,0,\ldots,0]^t$. Hence, the
components $(T^j\mathbf{r},T_{i}\mathbf{1})$ of the
vector \(\mathbf{b}_{i}\) can be extracted directly from the vector
$T^j\mathbf{r}$.  By solving the corresponding linear system we can
recover \(h_{i}(x_{n})\) for all \(i\in \mathcal{L}^c\).

Now we can easily recover the other univariate
polynomials \(h_{i}(x_{n})\) for all \(i\in \mathcal{L}\): by
definition of \(\mathcal{L}\) we have
$$
l_{i}=x_{i} +  \sum_{j\in \mathcal{L}^c} \alpha_{i,j}
  x_{j} + \alpha_{i,n}x_{n} + \alpha_{i,0} \in \mathbb{L}\subset \gdrl\, \text{ with } \alpha_{i,j}\in \mathbb{K}.
$$
Hence, the corresponding univariate polynomial $h_{i}(x_n)$ is simply
computed by the formula:
\[
h_{i}(x_n) = - \sum_{j\in \mathcal{L}^c} \alpha_{i,j} h_{j} (x_n)
- \alpha_{i,n} h_n(x_n) - \alpha_{i,0}\,.
\] Thus, we have reduced the number of linear systems  ~\eqref{hankelsystem} to solve from \(n-1\) to 
$n-\#\mathcal{L}-1$.

We conclude this section by summarizing the algorithm to compute
univariate polynomial representation in Algorithm~\ref{algchgord}.
For a deterministic version of Algorithm~\ref{algchgord}, we refer the
reader to Remark~\ref{rem:deterministic}.  In the next section, we
discuss how to compute the multiplication matrices.

\begin{algorithm}[!ht]
\SetKwData{Left}{left}\SetKwData{This}{this}\SetKwData{Up}{up}
\SetKwFunction{Union}{Union}\SetKwFunction{FindCompress}{FindCompress}
\SetKwInOut{Input}{Input}\SetKwInOut{Output}{Output}
  \Input{The multiplication matrix $T_n$ and the DRL Gröbner basis
  $\gdrl$ of an ideal $\id{I}$.}  \Output{Return the LEX Gröbner basis
  $\glex$ of $\id{I}$ or \textit{fail}.}  \caption{Univariate
  polynomial representation}\label{algchgord} Compute \(T^{2^{i}}\)
  for \(i=0,\ldots,\log_2D\) and compute \(T^{j}\mathbf{r}\)
  for \(j=0,\ldots,(2D-1)\) using induction \eqref{indu}. Deduce the
  linearly recurrent sequence $S$ and the Hankel
  matrix \(\mathcal{H}\) \; $h_{n}(x_n)
        := \text{BerlekampMassey}(S)$ \; \If{$\deg(h_{n}) = D$}{ Let
  $\mathcal{L}^c=\{j\in \{1,\ldots,n-1\}\text{ such that } \NF{drl}{x_{j}} = x_{j}\}$ and \(\mathcal{L}=\{1,\ldots,n-1\}\backslash\mathcal{L}^c\)\;
\For{$j\in \mathcal{L}^c$}{ Deduce \(T_j \mathbf{1}\) and \(\mathbf{b}_{j} \) then
  solve the structured linear system
  $\mathcal{H} \,\mathbf{c}_{j}=\mathbf{b}_{j} $\;
  $h_{j}(x_{n}):= \sum_{i=0}^{D-1}c_{j,i}x_n^i$ where \(c_{j,i}\) is
  the \(i\)th component of the vector \(\mathbf{c}_{j}\)\; } \For{$j\in \mathcal{L}$}{ $h_{j}(x_{n}) := - \sum_{i\in \mathcal{L}^c} \alpha_{j,i}
  h_{i}(x_{n}) - \alpha_{j,n}h_{n}(x_{n}) - \alpha_{j,0}$ where
  $\alpha_{j,i}$ is the \(i\)th coefficient of the linear form whose
  leading term is $x_{j}$\; } \Return $[x_1 -
  h_1(x_{n}),\ldots,x_{n-1}-h_{n-1}(x_{n}),h_{n}(x_n)]$\;
  } \textbf{else} \Return \textit{fail}\;
\end{algorithm}

\section{Multiplication matrices}\label{sec:matrix}


\subsection{The original algorithm in \(O(nD^3)\)}
To compute the multiplication matrices, we need to perform the
computation of the normal forms of all monomials $\epsilon_i x_j$ where
$1 \leq i \leq D$ and $ 1
\leq j \leq n$. 
\begin{proposition}[\cite{FGLM93}]\label{types}
Let $F = \{\epsilon_i x_j\ |\ 1 \leq i \leq D, 1 \leq
j \leq n\} \setminus B$ be the frontier of the ideal. Let $t= \epsilon_i x_j \in F$ then
\begin{enumerate}[I.]
\item\label{type1} either $t = \LT{drl}{g}$ for some $g \in \gdrl$ hence,
  $\NF{drl}{t} = t - g$;
\item\label{type2} or $t = x_k \,t'$ with $t' \in F$ and \(\deg(t')<\deg(t)\). Hence, if
  $\NF{drl}{t'} = \sum_{l=1}^s \alpha_l \epsilon_l$ with 
  \(\epsilon_s <_{\text{drl}} t'\), $\NF{drl}{t} = \NF{drl}{x_k \NF{drl}{t'}}=
   \sum_{l=1}^s \alpha_l \NF{drl}{\epsilon_l x_k}$.
\end{enumerate}
\end{proposition}

From this proposition, it is not difficult to see that the normal form
of all the monomials $\epsilon_i x_j$ can be easily computed if we
consider them in increasing order. Indeed, let $t = \epsilon_i x_j$
for some $i \in \{1,\ldots,D\}$ and $j \in \{1,\ldots,n\}$. Assume
that we have already computed the normal form of all monomials less
than $t$ and of the form $\epsilon_{i'} x_{j'}$. If $t$ is in $B$ or
is a leading term of a polynomial in $\gdrl$ then its normal form is
trivially known. If $t$ is of type~\eqref{type2} of
Proposition~\ref{types} then $t = x_k t'$ with $t' <_{\text{drl}} t$
hence $\NF{drl}{t'} = \sum_{i=l}^s \alpha_l \epsilon_l$ is
known. Finally, $\NF{drl}{t} = \sum_{l=1}^s \alpha_l \NF{drl}{x_k
  \epsilon_l}$ with $x_k \epsilon_l <_{\text{drl}} x_k t' = t$ for all
$l = 1,\ldots,s$. Thus the normal forms of $x_k \epsilon_l$ are known
for all $l = 1,\ldots,s$ and we can compute $\NF{drl}{t}$ in \(D^2\)
arithmetic operations. This yields the algorithm proposed in
\cite{FGLM93}. However, since the cardinal of the frontier \(F\) can
be bounded by \(n\,D\) the overall complexity is \(O(n D^3)\)
arithmetic operations.

\subsection{Computing the multiplication matrices using fast linear algebra}
Another way to compute the normal form of a term $t$ is to find the
unique polynomial in the ideal whose leading term is $t$ and the
other terms correspond to monomials in $B$. Hence, to compute the
multiplication matrices, we look for the polynomial $t - \NF{drl}{t}$
for any $t$ in the frontier $F$ (see
Proposition~\ref{types}). Therefore, to compute these polynomials we
proceed in two steps. First, we construct a polynomial in the ideal
whose leading term is $t$. If $t$ is the leading term of a polynomial
$g$ in $\gdrl$ then the desired polynomial is $g$ itself. Otherwise,
$t$ is of type~\ref{type2} of Proposition~\ref{types} and \(t =
x_kt'\) with \(t'\in F\) and \(\deg(t')<\deg(t)\). We will proceed
degree by degree so that we can assume we know a polynomial $f'$ in
the ideal whose leading term is $t'$; then the desired polynomial is
$f = x_kf'$. Next, once we have all the polynomials $f$ with all
possible leading terms $t$ of some degree \(d\), we can recover the
canonical form $t - \NF{drl}{t}$ by reducing $f$ with respect to the
other polynomials whose leading terms are less than $t$. By computing
a reduced row echelon form of the Macaulay matrix (the matrix
representation) of all these polynomials, we can reduced all of them
simultaneously.

Following the idea presented above, we can now describe
Algorithm~\ref{algmultmat} for computing all the multiplication
matrices \(T_{i}\). Assuming that \(F\) is sorted in increasing order
w.r.t. \(<_{drl}\), we define the linear map $\phi$:
$$
\phi :
\left(\begin{array}{cccc}
 & \A & \rightarrow & \K^{D+\#F}\\
& \sum_{i=1}^D \alpha_i \epsilon_i  + \sum_{j=1}^{\#F} \beta_j t_j & \mapsto & (\beta_{\#F},\ldots,\beta_1,\alpha_1,\ldots,\alpha_D)\,.
\end{array}\right)
$$ Let $M$ be a row indexed matrix by all the monomials in $F$. Let
$m$ be a monomial in $F$ and $i$ the position of $m$ in $F$, $M[m]$
denotes the row of $M$ of index $m$ \textit{i.e.} the $(\#F
-i+1)^{\text{th}}$ row of $M$ containing a polynomial of leading term
$m$.  If $T$ is a matrix, $T[*,i]$ denotes the $i^{\text{th}}$ column
of $T$.

\begin{algorithm}[!ht]
\SetKwData{Left}{left}\SetKwData{This}{this}\SetKwData{Up}{up}
\SetKwFunction{Union}{Union}\SetKwFunction{FindCompress}{FindCompress}
\SetKwInOut{Input}{Input}\SetKwInOut{Output}{Output}
 \Input{The DRL Gröbner basis $\gdrl$ of an ideal $\id{I}$.}
  \Output{The $n$ multiplication matrices $T_1,\ldots,T_n$.} 
  \caption{Building multiplication matrices
   (in the following  $||$ does not mean
parallel code but gives details 
about pseudo code on the left side).}
  \label{algmultmat}
  Compute $B = \{\epsilon_1 < \cdots < \epsilon_D\}$ and $F = \{ x_i \epsilon_j\ |\ i = 1,\ldots,n \text{ and } j = 1,\ldots,D\} \setminus B$, \(S:=\#F\)\;
  $d_{\min} := \min(\{\deg(t)\ |\ t \in F\})$; $d_{\max} := \max(\{\deg(t)\ |\ t \in F\})$; $\text{NF} := []$\; 
  \(M:=\) the zero matrix of size $nD \times (n+1)D$ row indexed by all the  monomials in \(F\)\;
 \For{$d = d_{\min}$ {\bf to} $d_{\max}$}{\nllabel{loop1}
 \(F_d := Sort(\{ t \in F\ |\ \deg(t) = d\},<_{drl})\) \;
  \For{$m\in F_d$ }{\nllabel{loop2}
  \hspace*{-4mm}
  \begin{minipage}{0.43\textwidth}
     {\footnotesize Check if we can find:\\
     \hspace*{4mm}$(i) g \in \gdrl$ 
     such that $\LT{drl}{g} = m$}\\
{\footnotesize \hspace*{4mm}(ii)  \(t'\in F\) such that \(m=x_k t'\)}\\
     {\footnotesize Add the corresponding row to the matrix  $M$\;}
  \end{minipage}
  \vrule \hspace*{0.1pt}
  \vrule \hspace*{1mm}
  \begin{minipage}{0.43\textwidth}
   {\footnotesize\If{$m = \LT{drl}{g}$}{
  $M[m] := \phi(g)$\;
  }              
  \Else{
  Find $x_k$ and $t' \in F_{d-1}$ such that $m = x_kt'$\;
  $M[m] := \phi(m - x_k\text{NF}[t'])$\;
  }}
  \end{minipage}
     }
     ${M} := \text{ReducedRowEchelonForm}(M)$ \;
     \nllabel{rowech}
     \For{$i = 1$ {\bf to} $s_d$}{\nllabel{loopNF}
   \begin{minipage}{0.43\textwidth}
     {\footnotesize Read $\NF{drl}{m}$ from \({M}\)\;}
  \end{minipage} \hspace*{-4mm}
  \vrule \hspace*{0.1pt}
  \vrule \hspace*{1mm}
  \begin{minipage}{0.43\textwidth}
   {\footnotesize $\text{NF}[m] :=- \sum_{j=1}^D {M}[m,S+j]\, \epsilon_j$\;}
  \end{minipage}
     }    
  }
    \begin{minipage}{0.48\textwidth}
     { Construct $T_1,\ldots,T_n$ from NF\;}
  \end{minipage}
  \vrule \hspace*{0.1pt}
  \vrule \hspace*{1.5mm}
  \begin{minipage}{0.43\textwidth}
   {\footnotesize
   \textbf{for} $\epsilon$ \textit{in} $B$ \textbf{do} $\text{NF}[\epsilon] := \epsilon$\;
   \For{$t$ in $F \cup B$}{
    \For{$x_i$ s.t. $x_i$ divides $t$ and $\frac{t}{x_i} = \epsilon_j \in B$}{
      $T_i[*,j] := \psi(\text{NF}[t])$\;
      }
    }}
  \end{minipage}
  \Return $T_1,\ldots,T_n$\;
\end{algorithm}

\begin{proposition}Algorithm~\ref{algmultmat} is correct.
\label{proofAlgoMatMult}
\end{proposition}

\begin{proof}
The key point of the algorithm is to ensure that for each monomial in $F$
its normal form is computed and stored in NF before we use it.
We will prove the following loop invariant for all \(d\) in \(\{d_{\min},\ldots, d_{\max}\}\).
 
\textit{Loop invariant: at the end of step $d$, all the normal forms
  of the monomials of degree $d$ in the frontier $F$ are computed and are stored in
  NF. Moreover, the $m^{th}$ row of the matrix ${M}$ contains $\phi(m -
  \NF{drl}{m})$ for any monomial \(m\in F_{d}\).}

First, we assume that $d = d_{\min}$. Then, each monomial $t$ of degree
$d$ in $F$ is of type \eqref{type1} of
Proposition~\ref{types}. Indeed, if $t$ was of type \eqref{type2} then
there exists $t'$ in $F$ of degree $d-1$ which divides \(t\). This is
impossible because \(t'\in F_{d_{\min}-1}=\emptyset\). Hence, the
normal form of $t$ for $t\in F_{d_{\min}}$, is known and $M[t]$
contains $\phi(g)$ with \(g\) the unique element of $\gdrl$ such that
$\LT{drl}{g} = t$. Hence, \(M[t]=\phi(g)=\phi(t -
\NF{drl}{t})\). Moreover, since \(\gdrl\) is a reduced Gr\"obner basis
, the matrix $M$ is already in reduced row echelon form. Thus, the
loop in Line~\ref{loopNF} updates NF\([t]\) for all $t\in F_{d}$.

Let $d > d_{\min}$, we now assume that the loop invariant is true for
any degree less than $d$. For all \(t\in F_d\) the $t^{\text{th}}$ row
of $M$ contains either $\phi(t - \NF{drl}{t})$ if $t$ is of
type~\eqref{type1} or \(\phi(t - x_k \text{NF}[t'])\) if $t$ is of
type~\eqref{type2}. Since $\deg(t') = d-1$, by induction its normal
form is known and in NF. Hence \(\text{NF}[t']=\NF{drl}{t'}\) and \(
M[t]= \phi(x_k(t'- \NF{drl}{t'})\).  A first consequence is that,
before Line~\ref{rowech}, since we sort $F_d$ at each step, $M$ is an
upper triangular matrix with $M[t,t] = 1$ for all \(t\in F_{d}\), see
Figure~\ref{fig:M}. Note that sorting $F_d$ is required only to obtain
this triangular form. Let \(f\) be the polynomial \(\NF{drl}{t'}\).
Writing \(f=\sum_{j=1}^{D} \lambda_j \epsilon_j\) we have that
\(\lambda_j=0\) if \(\deg(\epsilon_j)\geq d\) since
\(\deg(\NF{drl}{t'})\leq\deg(t')=d-1\). So that \(f=\sum_{j=1}^{k}
\lambda_j \epsilon_j\) such that \(\deg(\epsilon_j)< d\) when \(j\leq
k\). Now for all \(j\) such that \(1\leq j\leq k\) we are in one of
the following cases:\begin{enumerate}
\item \(x_{k} \epsilon_{k}\in B\) so that \(\NF{drl}{x_{k}
  \epsilon_{k}}=x_{k} \epsilon_{k}\) is already reduced.
\item \(x_{k} \epsilon_{k}\in F\). Since \(d'=\deg(x_k \epsilon_k)\leq
  d\) it implies that \(x_{k} \epsilon_{k}\in F_{d'}\) so that the row
  \(M[x_{k} \epsilon_{k}]\) has been added to \(M\).
\end{enumerate} 

Moreover, since each row of the matrix $M$ contains polynomial in the
ideal $\langle \gdrl \rangle$ after the computation of the row echelon
form, the rows of the matrix $M$ contain also polynomials in $\langle
\gdrl \rangle$ being linear combination of the previous polynomials.
Hence, after the computation of the row echelon form of $M$, the row
\(M[t]\) is equal to \(\phi(t - \NF{drl}{t})\).

By induction, this finishes the proof of the loop invariant and then of the
correctness of Algorithm~\ref{algmultmat}.
\end{proof}

\section{Polynomial equations with fixed degree: the tame case}\label{sec:tame}


The purpose of this section, is to analyze the asymptotic complexity
of Algorithm~\ref{algchgord} and Algorithm~\ref{algmultmat} when the
degrees of the equations of the input system are uniformly bounded by
a fixed integer $d>1$ and to establish the first main result of
this paper.

\subsection{General Complexity analysis}

We first analyse Algorithm~\ref{algchgord} to compute the univariate
polynomial representation given the last multiplication matrix.

\begin{proposition}
\label{prop:probaalgo}
Given the multiplication matrix $T_n$ and the DRL Gröbner basis
$\gdrl$ of an ideal in \textit{Shape Position}, its LEX Gröbner basis
can be probabilistically computed in $O(\log_2(D)(D^{\omega} +
n\log_2(D)D))$ where \(D\) is the number of solutions. Expressed with
the input parameters of the system to solve the complexity is
$O(nd^{\omega n})$ where \(d>1\) is a (fixed) bound on the degree of the
input polynomials.
\end{proposition}

\begin{proof}
As usual $T = T_n^t$ is the transpose matrix of $T_n$. Using the
induction~\eqref{indu}, the vectors \(T^{j}\mathbf{r}\) can be
computed for all \(j=0,\ldots,(2D-1)\) in $O(\log_2(D)D^\omega)$ field
operations. Then the linear recurrent sequence \(S\) and the matrix
$\mathcal{H}$ can be deduced with no cost.  The Berlekamp-Massey
algorithm compute the minimal polynomial of $S$ in $O(D\log_2^2(D))$
field operations~\cite{JoMa89,BGY80}.

As defined in Section~\ref{sec:other polys},
$\mathcal{L}=\{j\in \{1,\ldots,n-1\}\text{ such that
} \NF{drl}{x_{j}} \neq\ x_{j}\}$
and \(\mathcal{L}^c =\{1,\ldots,n-1\}\backslash\mathcal{L}\). The right
hand sides of the linear systems $\mathbf{b}_{i}$ can be computed
without field operations when \(i\in \mathcal{L}^c\). Since the matrix
$\mathcal{H}$ is a non singular Hankel matrix, the $\#\mathcal{L}^c$
linear systems \eqref{hankelsystem} can be solved in
$O(\#\mathcal{L}^c\log_2^2(D)D)=O(n\log_2^2(D)D)$ field operations.
Then, to recover all the $h_{i}(x_n)$ for $i\in \mathcal{L}$ we
perform $O(\#\mathcal{L}\#\mathcal{L}^c D)=O(n^{2}D)$ multiplications
and additions in $\K$.

Since the Bézout bound allows to bound $D$ by $d^n$ with $d$ a fixed
integer we have $\log_2(D) \leq n \log_2(d)$ and the arithmetic complexity
of Algorithm~\ref{algchgord} is $O(\log_2(D)(D^{\omega} +
n\log_2(D)D))$ which can be expressed in terms of $d$ and $n$ as
$O(nd^{\omega n})$.
\end{proof}

Note that the deterministic version, mentioned in
Remark~\ref{rem:deterministic} have a complexity in
$O(\log_2(D)D^\omega + D^2(n+\log_2(D)\log_2(\log_2(D))))$ arithmetic
operations, thanks to induction~\eqref{indu} and section 3.2.2
in~\cite{SparseFGLM}.  This deterministic version computes the LEX
Gröbner basis of the radical of the ideal in input when the ideal is
in \textit{Shape Position},.  In our case, this is not restricting
since in Problem~\ref{pb:posso} we assume that all the roots of the
system are simple which is equivalent to say that the ideal generated
by the polynomial is radical.

\begin{proposition}
\label{cor:det}
Let $T_n$ be the multiplication matrix and $\gdrl$ be the DRL
Gröbner basis of a radical ideal $\id{I}$ in \textit{Shape
Position}. There is a deterministic algorithm which computes the LEX
Gröbner basis of $\id{I}$ in $O(\log_2(D)D^\omega +
D^2(n+\log_2(D)\log_2(\log_2(D))))$ (or in $O(nd^{\omega n})$)
arithmetic operations in $\K$.
\end{proposition}

Now, to complete the first algorithm, we deal with the complexity of
Algorithm~\ref{algmultmat} to compute the multiplication matrices.
Note that in proposition~\ref{prop:probaalgo} and~\ref{cor:det} only
the last matrix \(T_n\) is needed. Before to consider the complexity
of Algorithm~\ref{algmultmat}, we first discuss about the complexity of computing $B$ and $F$.

\begin{lemma}\label{lem:BF}
Given $\gdrl$ (resp. $B$) the construction of $B$ (resp. $F$) requires
at most $O(n^3D^2)$ (resp. $O(nD^2+n^2D)$) elementary operations which
can be decreased to $O(nD)$ (resp. $O(n^2D)$) elementary operations if a 
hash table is used. 
\end{lemma}

\begin{proof}
It is well known that the canonical basis $B$ can be computed in
polynomial time (but no arithmetic operations). Nevertheless, in order
to be self contained we describe an elementary algorithm to
compute \(B\). We start with the monomial \(1\) and we multiply it by
all the variables \(x_{i}\) which gives $n$ new monomials to
consider. If the new monomials are not divisible by a leading term of
a polynomial in $\gdrl$ then we keep it otherwise we discard it. At
each step we multiply by the variables \(x_{i}\) only the monomials of
highest degree that we have kept and we proceed until the step where
all the new monomials are discarded. Hence, we have to test the
irreducibility of all the elements in \(F\cup B\) whose total number
is bounded by $(n+1)D$. Since \(\LT{drl}{\gdrl}\subset F\) we can
bound the number of elements of \(\gdrl\) by \(nD\). Therefore, to
compute $B$ we have to test the divisibility of $(n+1)D$ monomials by
at most $nD$ monomials. Hence, the construction of $B$ can be done in
$O(n^{3}D^{2})$ elementary operations. Note that by using a hash table
and assuming we have no memory limit, for each monomial we can test
its divisibility by a leading term of polynomials in $\gdrl$ in $O(1)$
operations. In that case $B$ can be constructed in $O(nD)$
elementary operations.

From $B$, the construction of $F$ requires $nD$ monomials
multiplications \textit{i.e.} $n^2D$ additions of integers. Moreover,
removing $B$ of $F$ can be done by testing if $(n+1)D$ monomials are
in $B$ in at most $O(nD^2)$ elementary operations which can be
decreased to $O(nD)$ if we use a hash table.
\end{proof}

Now we seen how constructing $B$ and $F$, the complexity of
Algorithm~\ref{algmultmat} is treated in the following proposition.

\begin{proposition}\label{prop:mult}
Given the DRL Gröbner basis $\gdrl$ of an ideal, we can compute all
the multiplication matrices in $O((d_{\max}-d_{\min})n^\omega
D^\omega)$ (or in $O((d_{\max}-d_{\min})n^\omega d^{\omega n})$)
arithmetic operations in $\K$ where \(d_{\max}\) (resp. \(d_{\min}\))
is the maximal (resp. the minimal) degree of all the polynomials
in \(\gdrl\).
\end{proposition}

\begin{proof}
Algorithm~\ref{algmultmat}, computes all the multiplication matrices incrementally degree by degree.
\noindent The frontier \(F\) can be written as the union of  disjoint sets \(F_{\delta}=\{ t \in F\ |\ \deg(t) = \delta\}\) so that we  define \(s_\delta:=\#F_\delta\) and \(S_\delta :=s_{d_{\min}}+\cdots+s_\delta\).
The cost of the loop at Line~\ref{loop1} is, at each step, given by
the complexity of computing the reduced row echelon form of $M$. In
degree \(\delta\) the shape of the matrix $M$ is depicted on
Figure~\ref{fig:M} where \(\textbf{Id}(S_{\delta-1})\) is
the \(S_{\delta-1}\times S_{\delta-1}\) identity matrix,
\(\textbf{0}(S_{\delta-1})\) is the \(S_{\delta-1}\times s_{\delta}\)
 zero matrix, \(T\) is a \(s_{\delta}\times s_\delta\) upper triangular matrix
 and \(B,C,D\) are dense matrices of respective size \(s_\delta\times
 S_{\delta-1}\), \(s_\delta\times D\), \(S_{\delta-1}\times D\).

\begin{figure}[!h]
{\scriptsize
$$
M=\begin{array}{r|cccc|ccc|ccc|}
\multicolumn{1}{c}{} & 
\multicolumn{4}{c}{t\in F_\delta} 
     && t\in F_{\delta-1}\cup\cdots\cup F_{d_{\min}}  & \multicolumn{1}{c}{} & 
     \multicolumn{1}{c}{} &t\in B & \multicolumn{1}{c}{}\\
\cline{2-11}
 & 1 & \star & \cdots & \star & \star & \cdots & \star & \star & \cdots & \star\\
 & 0 & 1 & \cdots & \star & \star & \cdots & \star & \star & \cdots & \star\\
 & \vdots & \textcolor{red}{\textbf{T}} & \ddots & \vdots & \vdots &  \textcolor{red}{\textbf{B}} & \vdots & \vdots & \textcolor{red}{\textbf{C}} & \vdots\\
 & 0 & 0 & \cdots & 1 & \star & \cdots & \star & \star & \cdots & \star\\
\cline{2-11}
 & 0 & 0 & \cdots & 0 & 1 & \cdots & 0 & \star & \cdots & \star\\
 & \vdots & \multicolumn{2}{c}{\textcolor{red}{\textbf{0($S_{\delta-1},s_\delta$)}}} & \vdots & \multicolumn{1}{c}{\textcolor{red}{\textbf{Id($S_{\delta-1}$)}}} & \ddots & & \vdots & \textcolor{red}{\textbf{D}} & \vdots\\
 & 0 & 0 & \cdots & 0 & 0 & \cdots & 1 & \star & \cdots & \star\\
\cline{2-11}
\end{array}
$$

}
\caption{\label{fig:M}Shape of the matrix $M$ of Algorithm~\ref{algmultmat}.\\
}
\end{figure}
Consequently the  reduced row echelon form of \(M\) can be obtained from the following formula:
$$
\text{ReducedRowEchelonForm}(M) = \left [
\begin{array}{c|c}
 & T^{-1}(C - BD)\\
\textbf{Id}(S_\delta) & -------\\
& D
\end{array} \right ]\,.
$$ Since \(s_{\delta}\leq S_\delta\leq S_{d_{\max}}\leq n D\) we can
bound the complexity of computing the reduced row echelon form of $M$
by $O(n^\omega D^\omega)$. From Lemma~\ref{lem:BF}, the costs of the
construction of $B$ and $F$ are negligible in comparison to the cost
of loop in Line~\ref{loop1} which therefore gives the complexity of
Algorithm~\ref{algmultmat}: $O((d_{\max}-d_{\min})n^\omega D^\omega))$
arithmetic operations. Since $D \leq d^n$, this complexity can be
written as $O((d_{\max}-d_{\min})n^\omega d^{\omega n})$.
\end{proof}

\subsection{Complexity for regular systems}\label{sec:compreg}
Regular systems form an important family of polynomial
systems. Actually, the complexity of computing a Gr\"obner basis of a regular system is well understood. Since  the property of being
regular is a generic property this also the typical behavior  of polynomial systems.

\begin{definition}\label{def:regular}
A sequence of non zero homogeneous polynomials $(f_1,\ldots,f_m) \in
\A^m$ is regular if for all $i = 1,\ldots,m-1$, $f_{i+1}$ does not
divide $0$ in $\A/\Id{f_1,\ldots,f_i}$. A sequence of non zero affine
polynomials is regular if the sequence $(f_1^h,\ldots,f_m^h)$ is
regular where $f_i^h$ is the homogeneous part of highest degree of
$f_i$.
\end{definition}

For  regular systems we can bound accurately the  values of  $d_{\max}$ which is the maximal degree of \(\gdrl\) and we can prove the first
main result of this paper.

\begin{theorem}\label{thm:main1}
Let $\mathcal{S} = \{ f_1,\ldots,f_n \}$ be a polynomial system
generating a radical ideal admitting a LEX Gröbner basis in
\textit{Shape Position}. Assume that $(f_1,\ldots,f_n)$ is a regular
sequence of polynomials whose degrees are uniformly bounded by a fixed
integer $d$ \textit{i.e.} $\deg(f_i) \leq d$ for $i = 1,\ldots,n$. The
univariate polynomial representation of all the solutions of
$\mathcal{S}$ can be computed using a deterministic algorithm in
$O(nd^{\omega n} + (dn^{\omega+1}+\log_2(D))D^\omega)$ arithmetic
operations in $\K$.
\end{theorem}

\begin{proof}
For regular systems \(d_{\max}\) can be bounded by the Macaulay bound~\cite{Laz83,BFSY05}:  $d_{\max} \leq
\sum_{i=1}^n(\deg(f_i) - 1) + 1 \leq n(d-1)+1$ .  Given the system $\mathcal{S}$ the complexity of
computing the DRL Gröbner basis of $\langle S \rangle$ is bounded by~\cite{BFSY05}:
$$O\left ( n\binom{n + d_{\max}}{n}^\omega \right ) = O\left (
n\binom{nd + 1}{n}^\omega \right ) = O(nd^{\omega n})$$ arithmetic
operations.

From this DRL Gröbner basis, according to Proposition~\ref{prop:mult},
the multiplication matrix $T_n$ can be computed in
$O(dn^{\omega+1}D^\omega)$ arithmetic operations.

Finally, from $T_n$ and the DRL Gröbner basis, thanks to
Proposition~\ref{cor:det} the univariate polynomial representation can
be computed by a deterministic algorithm in $O(\log_2(D)D^\omega +
D^2(n+\log_2(D)\log_2(\log_2(D))))$ arithmetic operations.  Since,
$F_4$ \cite{Fau99}, $F_5$ \cite{Fau02} and Algorithm~\ref{algmultmat}
are deterministic algorithms this finishes the proof.
\end{proof}

In particular, a generic system is regular. Let $d_i
= \deg(f_i)$ for all $i=1,\ldots,n$.  Since the Bézout bound allows
to bound the number of solutions $D$ by $\prod_{i=1}^n d_i \leq d^n$
and since this bound is generically reached, we have generically that
$D = \prod_{i=1}^n d_i \leq d^n$ and we get the following corollary.

\begin{corollary}
Let \(\K\) be the rational field $\Q$ or a finite field
$\F_q$. Let \(\mathcal{S} = \{f_1,\ldots,f_n\}\subset
\K[x_1,\ldots,x_n]\) be a generic polynomial system generating an ideal $\id{I} = \langle\mathcal{S}\rangle$ of degree \(D\). If $\id{I}$ admits a LEX Gröbner basis in \textit{Shape Position} and if the degree of each polynomial in $\mathcal{S}$ is uniformly
bounded by a fixed integer $d$ then there exists a deterministic
algorithm which computes the univariate polynomial representation of
the roots of \(\mathcal{S}\) in $\widetilde{O}(d^{\omega
n})=\widetilde{O}(D^\omega)$ arithmetic operations where the notation
$\widetilde{O}$ means that we neglect logarithmic factors in $D$ and
polynomial factors in $n$.
\end{corollary}

In the next section, we study a first step towards the generalization
of Theorem~\ref{thm:main1} to polynomial systems with equations of non
fixed degree. More precisely, we are going to discuss what happens if
one polynomial have a non fixed degree \textit{i.e.} its degree
depends on a parameter (for instance the number of variables).
In this case, Theorem 5.1 does not apply but we present other
arguments in order to obtain a similar complexity results for
computing $\glex$ given $\gdrl$ and new ideas for its generalization.

\section{A worst case ultimately not so bad}
\label{sec:worst}

We consider, for instance, the following pathological case: $\deg(h_1)
= \cdots = \deg(h_{n-1}) = 2$ and $\deg(h_n) = 2^n$. Then, $D =
2^{2n-1}$, $d_{\min} = 2$ and $d_{\max} = 2^n +n -1$. In this context,
the complexity of computing $\glex$ given $\gdrl$ seems to be in
$O( \log_2^\omega(D)D^{\omega+\frac{1}{2}})$ arithmetic
operations. However, we will show that an adaptation of
Algorithm~\ref{algmultmat} allows to decrease this complexity.

In \cite{Mor03}, Moreno-Socias studied the basis of the residue class
ring $\A/\id{I}$, w.r.t. the DRL ordering, for generic ideals.  In
particular, he shows that when the smallest variable $x_n$ is in
abscissa any section of the stairs of $\id{I}$ has steps of height one
and of depth two. That is to say, for any variable $x_i$ with $i<n$
and for all instantiations of the other variables
($\{x_1,\ldots,x_{n-1}\} \setminus \{x_i\}$) the associated section of
the stairs of $\id{I}$ has the shape in Figure~\ref{stairs}.

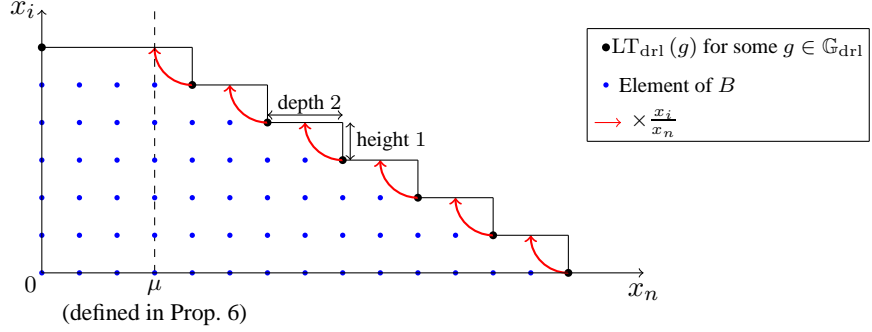
\begin{figure}[!h]
\begin{center}  
        \begin{tikzpicture}[x={(0.5cm,0cm)},y={(0cm,0.5cm)}]
        \fill[black] (0,6) circle (1.5pt);
        \fill[black] (4,5) circle (1.5pt);
        \draw[->,red,thick] (4,5) arc (-90:-180:1);
        \fill[black] (6,4) circle (1.5pt);
        \draw[->,red,thick] (6,4) arc (-90:-180:1);
        \fill[black] (8,3) circle (1.5pt);
        \draw[->,red,thick] (8,3) arc (-90:-180:1);
        \fill[black] (10,2) circle (1.5pt);
        \draw[->,red,thick] (10,2) arc (-90:-180:1);
        \fill[black] (12,1) circle (1.5pt);
        \draw[->,red,thick] (12,1) arc (-90:-180:1);
        \fill[black] (14,0) circle (1.5pt);
        \draw[->,red,thick] (14,0) arc (-90:-180:1);
        \node (x1) at (-0.5,7) {$x_i$};
        \node (x2) at (16,-0.5) {$x_n$};
        \node (O) at (-0.3,-0.3) {{\footnotesize $0$}};
        \node (mu) at (3,-0.7) {{\footnotesize \begin{tabular}{c}$\mu$ \\
        (defined in Prop.~\ref{Moreno})
        \end{tabular}}}; 
        \fill[black] (15,6) circle (1.5pt);       
        \node (gen) at (18.5,6) {{\scriptsize $\LT{drl}{g}$ for some $g \in \gdrl$}}; 
        \fill[blue] (15,5) circle (1pt);
        \node (basis) at (16.9,5) {{\scriptsize Element of $B$}};
        \draw[->,red] (14.7,4) -- (15.4,4);
        \node at (16.3,4) {{\scriptsize $\times \frac{x_i}{x_n}$}};
        \draw[-] (14.5,3.5) -- (14.5,6.5);
        \draw[-] (14.5,3.5) -- (22,3.5);
        \draw[-] (14.5,6.5) -- (22,6.5);
        \draw[-] (22,6.5) -- (22,3.5);
        \draw[dashed] (3,0) -- (3,7);
        \foreach \y in {0,1,...,5} {
                 \pgfmathsetmacro{\lastpoint}{13-2*\y}
                 \foreach \x in {0,1,...,\lastpoint} {
                 \fill[blue] (\x,\y) circle (1pt);
                 }
        }
        \draw[->] (0,0) -- (0,7);
        \draw[->] (0,0) -- (16,0);
        \draw[-] (0,6) -- (4,6);
        \draw[-] (4,6) -- (4,5);
      
        \draw[-] (4,5) -- (6,5);
        \draw[-] (6,5) -- (6,4);
      
        \draw[-] (6,4) -- (8,4);
        \draw[-] (8,4) -- (8,3);
      
        \draw[-] (8,3) -- (10,3);
        \draw[-] (10,3) -- (10,2);
      
        \draw[-] (10,2) -- (12,2);
        \draw[-] (12,2) -- (12,1);
      
        \draw[-] (12,1) -- (14,1);
        \draw[-] (14,1) -- (14,0);
      
       \draw[<->] (6,4.2) -- (8,4.2); 
       \node at (7.12,4.5) {{\scriptsize depth \(2\)}};
       \draw[<->] (8.2,3) -- (8.2,4); 
        \node at (9.3,3.5) {{\scriptsize height \(1\)}};
       \end{tikzpicture}       
\end{center}
        \caption{Section of the stairs of generic ideals with $\deg(x_j)$ fixed for all $j \in \{1,\ldots,n-1\}\setminus i$.\label{stairs}} 
\end{figure}

This shape is summarized in Proposition~\ref{Moreno}. 

\begin{proposition}[Moreno-Socias\cite{Mor03}]\label{Moreno}
Let $\widetilde{B}_i = \{ m = x^{\alpha_1}\cdots
x_{n-1}^{\alpha_{n-1}}\ |\ mx_n^i \in B \}$. Let $\delta
= \sum_{i=1}^{n}(\deg(h_i)-1)$, $\delta^*
= \sum_{i=1}^{n-1}(\deg(h_i)-1)$ and $\sigma = \min \left
(\delta^*, \lfloor \frac{\delta}{2} \rfloor \right )$. Let $\mu
= \delta - 2\sigma$, then 
\begin{enumerate}
\item $\widetilde{B}_0 = \cdots =
\widetilde{B}_\mu$ and $\widetilde{B}_i
= \widetilde{B}_{i+1}$ for $\mu < i < \delta$ and $i \not \equiv
\delta \mod 2$;
\item The leading term of polynomials in $\gdrl$ of degree $0$ in $x_n$ have degree at most $\sigma+1 = \bar{\sigma}$;
\item\label{leadterm} The leading term of polynomials in $\gdrl$ of degree $\alpha$ in $x_n$ with $\mu < \alpha \leq \delta +1 $ with $\alpha \not \equiv \delta \mod 2$ are all of total degree $d  + \alpha$ where $d = \max(\deg(m)\ |\ m \in \widetilde{B}_{\alpha-1})$. Moreover, all these leading terms are exactly given by $t = mx_n^\alpha$ for all $m \in \widetilde{B}_{\alpha-1}$ of degree $d$;
\item\label{notinG} There is no leading term of polynomials in $\gdrl$ of degree $1,\ldots,\mu$ in $x_n$ or of degree $\alpha$ in $x_n$ with $\alpha > \delta+1$ or $\mu \leq \alpha \leq \delta$ and $\alpha \equiv \delta \mod 2$.
\end{enumerate}
\end{proposition}

In our case, we have $d_{\max} = \delta +1$, $\delta^* = n-1$, $\delta
= 2^n +n -2$, $\sigma = n-1$ and $\mu = 2^n -n$. We can note that 
in this particular case, $\mu$
is very large which implies that a large part of the monomials of the
form $\epsilon_i x_j$ are actually in $B$.
We will show that in Algorithm~\ref{algmultmat} instead of computing
the loop in Line~\ref{loop1} for $d = d_{\min},\ldots,d_{\max}$ we can
perform it only on restricted subset $d =
d_{\min},\ldots,\sigma(n-1)+1, \mu+1,\ldots,d_{\max}$. By consequence,
the complexity of computing $\glex$ given $\gdrl$ will be in
$O((d_{\max} - \mu + \sigma(n-1) - d_{\min})n^\omega D^\omega) =
O(\log_2^{\omega+2}(D)D^\omega)$ with $d_{\max} - \mu + \sigma(n-1) -
d_{\min} = n^2 - 2 \sim \log_2^2(D)$.

\begin{lemma}
Given the normal form of all monomials in $F$ of degree less or equal to
$\sigma(n-1)+1$ we can compute all the normal forms of all monomials in
$F$ of degree less or equal than $\mu$ in less than $O(nD^2)$
arithmetic operations.
\end{lemma}

Suppose that we know the normal form of the monomials of the forms
$\epsilon_i x_j$ of degree less than $\mu$ which are not divisible by
$x_n$. From these normal forms, the idea of the proof is to show that
the normal form of all the monomials of the form $\epsilon_i x_j$ of
degree less than $\mu$ and of degree $\alpha_n > 0$ in $x_n$ is given by
$x_n^{\alpha_n}\NF{drl}{t}$ where $\NF{drl}{t}$ is assumed to be known.

\begin{proof}
Let $t \in F$ of degree less or equal to $\mu$. First, assume that
$x_n$ does not divide $t$. As $\id{I}$ is zero dimensional, there
exists $\eta_1,\ldots,\eta_{n-1} \in \N$ such that $x_i^{\eta_i}$ is a
leading term of a polynomial in $\gdrl$. Moreover, from
Proposition~\ref{Moreno}, $\eta_i \leq \bar{\sigma}$. Hence, for all
$\epsilon \in \widetilde{B}_0$, $\deg(\epsilon) \leq \sigma
(n-1)$. The monomials in $F$ not divisible by $x_n$ are all of the
form $x_i \epsilon$ with $i=1,\ldots,n-1$ and
$\epsilon \in \widetilde{B}_0$. Thus $\deg(t) \leq \sigma(n-1)+1$ and
by hypothesis, its normal form is known.

Suppose now that $x_n$ divides $t$ and $t$ is of type~\ref{type2} of
Proposition~\ref{types}. We can write $t = x_n^{\alpha} t'$ where
$\alpha \in \N^*$ such that $x_n \nmid t'$. From
Proposition~\ref{Moreno} item \eqref{notinG}, $t'$ is a leading term
of a polynomial in $\Id{\gdrl}$. Moreover, $t \in F$ so $t = x_i
\epsilon$ with $\epsilon \in B$. Suppose that $i = n$ hence,
$\frac{t}{x_n} = \epsilon = x_n^{\alpha-1} t' \in \Id{\gdrl}$ which is
impossible. Thus, $i \ne n$ and we have, $t' = \frac{t}{x_n^{\alpha}}
= x_i \epsilon'\in F$ with $\epsilon' = \frac{\epsilon}{x_n^\alpha}
\in B$. Therefore, from the first part of this proof, $\NF{drl}{t'} =
\sum_{i=1}^s \alpha_i \epsilon_i$, $\alpha_i \in \K$ is
known. Finally, $\NF{drl}{t} = \sum_{i=1}^s \alpha_i
\NF{drl}{x_n^\alpha \epsilon_i}$ with $\deg(x_n^\alpha \epsilon_i)
\leq \mu$. Let $k_i$ be such that $x_n^{k_i} | \epsilon_i$ and
$x_n^{k_i+1} \nmid \epsilon_i$ as $\widetilde{B}_{k_i} =
\widetilde{B}_{k_i+\alpha}$ then $x_n^\alpha \epsilon_i \in B$ and
$\NF{drl}{t} = \sum_{i=1}^s \alpha_i x_n^\alpha \epsilon_i$.  

By consequence, computing the normal form of $t$ can be done
in less than $D$ arithmetic operations. As usual, we can bound the size of $F$
by $nD$ which finishes the proof.
\end{proof}

One can notice that Algorithm~\ref{algchgord} -- which computes
univariate polynomial representation -- takes as input only the
multiplication matrix by the smallest variable. Thus in the proof of
Theorem~\ref{thm:main1} we did not fully take advantage of this
particularity. Hence, the next section is devoted to study if this
matrix can be computed more efficiently than computing all the
multiplication matrices. By studying the structure of the basis of the
$\K$-vector space $\A/\id{I}$ we will show that, up to a linear change
of variables, $T_n$ can be deduced from
$\gdrl$. In the previous results, the algorithm restricting the order
of magnitude of the degrees of the equations is
Algorithm~\ref{algmultmat} to compute the multiplication
matrices. Since, we need only $T_n$ which can be computed very
efficiently, the impact of such a result is that there exists a Las
Vegas algorithm extending the result of Theorem~\ref{thm:main1} to
polynomial systems whose equations have non fixed degree.

\section{Polynomial equations with non-fixed degree: the wild case}\label{sec:wild}


In this section, in order to obtain our main result, we
consider \emph{initial} and \emph{generic} ideals. The initial ideal
of $\id{I}$, denoted $\text{in}_{<}(\id{I})$, is defined by
$\text{in}_{<}(\id{I}) = \{ \text{LT}_{<}(f)\ |\ f \in \id{I} \}$. A minimal
set of generators of $\text{in}_{<}(\id{I})$ is denoted $\stair{\id{I}}$,
and is given by the leading terms of the polynomials in the Gröbner
basis of $\id{I}$ w.r.t. the monomial ordering $<$. To compute the
multiplication matrix $T_n$ we need to compute the normal forms of all
monomials $\epsilon_i x_n$ for $i=1,\ldots,D$ with $\epsilon_i \in
B$. As mentioned in Section~\ref{sec:matrix} a monomial of the form
$\epsilon_i x_n$ can be either in $B$ or in $\stair{\id{I}}$ or in
$\text{in}_{<}(\id{I}) \setminus \stair{\id{I}}$. As previously shown, the
difficulty to compute $T_n$ lies in the computation of the normal
forms of monomials $\epsilon_i x_n$ that are in
$\text{in}_{<}(\id{I}) \setminus \stair{\id{I}})$. In this section, thanks
to the study of the stairs, \textit{i.e.} $B$, of generic ideals by
Moreno-Socias, see Section~\ref{sec:worst}, we first show that for
generic ideals, \textit{i.e.}  ideals generated by generic systems (as
defined in Section~\ref{sec:compreg}), all monomials of the form
$\epsilon_i x_n$ are in $B$ or in $\stair{\id{I}})$. Hence, the
multiplication matrix $T_n$ can be computed very efficiently. Then, we
show that, up to a linear change of variables, this result can be
extended to any ideal. According to these results, we finally propose
an algorithm for solving the PoSSo problem whose complexity allows to
obtain the second main result of this paper.

\subsection{Reading directly $T_n$ from the Gr\"obner basis}\label{Tnfree}

In the sequel, the arithmetic
operations will be the addition or the multiplication of two operands
in $\K$ that are different from $\pm 1$ and $0$. In particular we do
not consider the change of sign as an arithmetic operation.

\begin{proposition}\label{height1}
Let $\id{I}$ be a generic ideal.  Let $t$ be a monomial in
$\stair{\id{I}}$ \textit{i.e.} a leading term of a polynomial in the DRL
Gröbner basis of $\id{I}$. If $x_n$ divides $t$ then for all $k \in 
\{1,\ldots,n-1\}$,  $\frac{x_kt}{x_n} \in \inId{drl}{\id{I}}$.
\end{proposition}

\begin{proof}
This result is deduced from the shape of the stairs of $\id{I}$ (see
Figure~\ref{stairs} for a representation in dimension $2$).  Let $t =
x_1^{\alpha_1}\cdots x_n^{\alpha_n}$ be a leading term of a polynomial
in $\gdrl$ divisible by $x_n$ \textit{i.e.} $\alpha_n > 0$ and $m =
x_1^{\alpha_1}\cdots x_{n-1}^{\alpha_{n-1}}$. We use the same
notations as in Proposition~\ref{Moreno}.

From Proposition~\ref{Moreno} item~\eqref{notinG}, since $t \in
E(\id{I})$ and $\alpha_n > 0$ we have $\alpha_n > \mu$ and
$\alpha_n \not \equiv \delta \mod 2$.
Then, from Proposition~\ref{Moreno} item \eqref{leadterm},
$\deg(m)$ is the maximal degree reached by the monomials in
$\widetilde{B}_{\alpha_{n-1}}$. Thus $x_k m \notin
\widetilde{B}_{\alpha_{n-1}}$ for all $k \in \{1,\ldots,n-1\}$. As a
consequence, for all $k \in \{1,\ldots,n-1\}$ we have
$\frac{x_kt}{x_n} \in \inId{drl}{\id{I}}$.
\end{proof}

Consequently, from the previous proposition, we obtain the following
result.

\begin{theorem}\label{thm:Tnopt1}
  Given $\gdrl$ the DRL Gröbner basis of a generic ideal $\id{I}$, the
  multiplication matrix $T_n$ can be read from \(\gdrl\) with no
  arithmetic operation.
\end{theorem}

\begin{proof}
Suppose that there exists $i \in \{1,\ldots,D\}$ such that $t = x_n
\epsilon_i$ is of type~\eqref{type2}. Hence, $t = m\, \LT{drl}{g}$ for
some $g \in \gdrl$ and $\deg(m) > 1$ with $x_n \nmid m$ (otherwise
$\epsilon_i \notin B$). Then, there exists $k \in \{1,\ldots,n-1\}$
such that $x_k \mid m$.  By consequence,  from
Proposition~\ref{height1}, we have $\epsilon_i
= \frac{m}{x_k} \cdot \frac{x_k \LT{drl}{g}}{x_n} \in \inId{drl}{\id{I}}$
which yields a contradiction. Thus, all monomials $t = x_n \epsilon_i$
are either in $B$ or in $E(\id{I})$ and their normal forms are known
and given either by $t$ (if $t \in B$) or by changing the sign of some
polynomial $g \in \gdrl$ and removing its leading term. Note that by
using a linked list representation (for instance), removing the
leading term of a polynomial does not require arithmetic operation.
\end{proof}

Thanks to the previous theorem, Algorithm~\ref{algchgord} can be used
to compute the LEX Gröbner basis of a generic ideal:

\begin{corollary}
Let \(\id{I}\) be a generic ideal in \textit{Shape Position}. From the
DRL Gröbner basis\(\gdrl\) of $\id{I}$, its LEX Gröbner basis $\glex$
can be computed in $O(\log_2(D)(D^\omega + n\log_2(D)D))$ arithmetic
operations with a probabilistic algorithm or $O(\log_2(D)D^\omega +
D^2(n+\log_2(D)\log_2(\log_2(D))))$ arithmetic operations with a
deterministic algorithm.
\end{corollary}

However, polynomial systems coming from applications are usually not
generic. Nevertheless, this difficulty can be bypassed by applying a
linear change of variables. Let $g \in \text{GL}(\K,n)$ the ideal
$g \cdot \id{I}$ is defined as follows $g \cdot \id{I} = \{ f(g \cdot
X)\ |\ f \in \id{I} \}$ where $X$ is the vector $[x_1,\ldots,x_n]$.
By studying the structure of the \emph{generic initial ideal} of
$\id{I}$ -- that is to say, the initial ideal of $g \cdot \id{I}$ for
a generic choice of $g$ -- we will show that results of
Proposition~\ref{height1} and Theorem~\ref{thm:Tnopt1} can be
generalized to non generic ideals, up to a random linear change of
variables. Indeed, in \cite{Gal73} Galligo shows that for the
characteristic zero fields, the generic initial ideal of any ideal
satisfies a more general property than
Proposition~\ref{height1}. Later, Pardue~\cite{Par94} extends this
result to the fields of positive characteristic.

\begin{definition}
Let $\K$ be an infinite field and $\id{I}$ be an homogeneous ideal of
 $\K[x_1,\ldots,x_n]$. There exists a Zariski open set
 $U \subset \text{GL}(\K,n)$ and a monomial ideal $\id{J}$ such that
 $\inId{drl}{g \cdot \id{I}} = \id{J}$ for all $g \in U$.  The generic
 initial ideal of $\id{I}$ is denoted $\text{Gin}(\id{I})$ and is
 defined by $\id{J}$.
\end{definition}

The next result, is a direct consequence of \cite{Gal73,BS87,Par94}
and summarized in \cite[p.351--358]{Eis95}. This result allows to
extend, up to a linear change of variables, Proposition~\ref{height1}
to non generic ideals.

\begin{theorem}\label{borel}
Let $\K$ be an infinite field of characteristic $p \geq 0$. Let
$\id{I}$ be an homogeneous ideal of $\K[x_1,\ldots,x_n]$ and $\id{J}
= \text{Gin}(\id{I})$. For the DRL ordering, for all generators $m$ of
$\id{J}$, if $x_i^t$ divides $m$ and $x_i^{t+1}$ does not divide $m$
then for all $j < i$, the monomial $\frac{x_j}{x_i}m$ is in $\id{J}$
if $t \not \equiv 0 \mod p$.
\end{theorem}

Let $f = \sum_{i=0}^d f_i$ be an affine polynomial of degree \(d\) of
$\A$ where $f_i$ is an homogeneous polynomial of degree $i$. The
homogeneous component of highest degree of $f$, denoted $f^{h}$, is
the homogeneous polynomial $f_d$.  Let $\id{I}$ be an affine
ideal \textit{i.e.} generated by a sequence of affine polynomials.  In
the next proposition we highlight an homogeneous ideal having the same
initial ideal than $\id{I}$. This allows to extend the result of
Theorem~\ref{borel} to affine ideals. 

\begin{proposition}\label{affine}
Let $\id{I} = \Id{f_1,\ldots,f_s}$ be an affine ideal. If
$(f_1,\ldots,f_s)$ is a regular sequence, then there exists a Zariski
open set $U_a \subset \text{GL}(\K,n)$ such that for all $g \in U_a$,
$\stair{g \cdot \id{I}} = \stair{\text{Gin} \left (\id{I}^h \right
)}$.
\end{proposition}

\begin{proof}
Let $f$ be a polynomial. We denote by $f^h$ the homogeneous component
of highest degree of $f$ and $f^a = f - f^h$. Let
$t \in \inId{drl}{\id{I}}$, there exists $f \in \id{I}$ such that
$\LT{drl}{f} = t$. Since, $f \in \id{I}$ and $(f_1^h,\ldots,f_s^h)$ is
assumed to be a regular sequence then there exist
$h_1,\ldots,h_s \in \K[x_1,\ldots,x_n]$ such that $f = \sum_{i=1}^s
h_if_i = \sum_{i=1}^s h_if_i^h + \sum_{i=1}^s h_if_i^a$ with
$\deg(h_if_i) \leq \deg(f)$ for all $i \in \{1,\ldots,s\}$ and there exists
$j \in \{1,\ldots,s\}$ such that $\deg(h_jf_j) = \deg(f)$.  By
consequence, $0 \ne \sum_{i=1}^s h_if_i^h \in \id{I}^h$
where \(\id{I}^h\) is the ideal generated by $\{f_1^h,\ldots,f_s^h\}$
and $\LT{drl}{f} = \LT{drl}{\sum_{i=1}^s h_if_i^h}$. Thus,
$\inId{drl}{\id{I}} \subset \inId{drl}{\id{I}^h}$. It is
straightforward that $\inId{drl}{\id{I}^h} \subset \inId{drl}{\id{I}}$
hence $\inId{drl}{\id{I}^h} = \inId{drl}{\id{I}}$.

For all $g \in \text{GL}(\K,n)$, since $g$ is invertible the sequence
$(g \cdot f_1,\ldots,g\cdot f_s)$ is also regular. Indeed, if there
exists $i \in \{1,\ldots,s\}$ such that $g \cdot f_i$ is a divisor of 
zero in $\K[x_1,\ldots,x_n]/\Id{g \cdot f_1,\ldots, g \cdot f_i}$ then
$f_i$ is a divisor of zero in
$\K[x_1,\ldots,x_n]/\Id{f_1,\ldots,f_i}$.
Hence,
$$
\inId{drl}{g \cdot \id{I}} = \inId{drl}{(g \cdot \id{I})^h}\,.
$$
Moreover, $g$ is a linear change of variables thus it preserves the
degree. Hence, for all $f \in \id{I}$, we have $(g \cdot f)^h
= g \cdot f^h$.  Finally, let $U_a$ be a Zariski open subset of
$\text{GL}(\K,n)$ such that for all $g \in U_a$, we have the equality
$\inId{drl}{g \cdot \id{I}^h} = \text{Gin}(\id{I}^h)$. Thus, for all
$g \in U_a$, we then have $\inId{drl}{g\cdot \id{I}}
=\inId{drl}{(g \cdot \id{I})^h}=\inId{drl}{g \cdot \id{I}^h}
= \text{Gin}(\id{I}^h)$.
\end{proof}

Hence, from the previous proposition, for a random linear change of
variables $g \in \text{GL}(\K,n)$ we have $\inId{drl}{g \cdot \id{I}}
= \text{Gin} \left (\id{I}^h\right )$. Thus from Theorem~\ref{borel},
for all generators $m$ of $\inId{drl}{g \cdot \id{I}}$ (\textit{i.e.}
$m$ is a leading term of a polynomial in the DRL Gröbner basis of
$g \cdot \id{I}$) if $x_n^t$ divides $m$ and $x_n^{t+1}$ does not
divide $m$ then for all $j < n$ we have
$\frac{x_j}{x_n}m \in \inId{drl}{g \cdot \id{I}}$ if $t \not \equiv
0 \mod p$. Therefore, in the same way as for generic ideals, the
multiplication matrix $T_n$ of $g \cdot \id{I}$ can be read from its
DRL Gröbner basis. This is summarized in the following corollary.

\begin{corollary}\label{cor:Tnopt2}
Let $\K$ be an infinite field of characteristic $p \geq 0$. Let
$\id{I}$ be a radical ideal of $\K[x_1,\ldots,x_n]$. There exists a
Zariski open subset \(U\) of $\text{GL}(\K,n)$ such that for all
$g \in U$, the arithmetic complexity of computing the multiplication
matrix by $x_n$ of $g \cdot \id{I}$ given its DRL Gröbner basis can be
done without arithmetic operation.  If $p > 0$ this is true only if
$\deg_{x_n}(m) \not \equiv 0 \mod p$ for all
$m \in \stair{g \cdot \id{I}}$. Consequently, under the same
hypotheses, computing the LEX Gröbner basis of $g \cdot \id{I}$ given
its DRL Gröbner basis can be bounded by $O(\log_2(D)(D^\omega +
n \log_2(D)D))$ arithmetic operations.
\end{corollary}

Following this result, we propose another algorithm for polynomial
systems solving. 

\subsection{Another algorithm for polynomial systems solving}\label{newposso}

Let $\mathcal{S} \subset \K[x_1,\ldots,x_n]$ be a polynomial system
 generating a radical ideal denoted $\id{I}$.  For any
 $g \in \text{GL}(\K,n)$, from the solutions of $g \cdot \id{I}$ one
 can easily recover the solutions of $\id{I}$. Let $U$ be the Zariski
 open subset of $\text{GL}(\K,n)$ such that for all $g \in U$,
 $\inId{drl}{g \cdot \id{I}} = \text{Gin}(\id{I}^h)$. If $g$ is chosen
 in $U$ then the multiplication matrix $T_n$ can be computed very
 efficiently. Indeed, from Section~\ref{Tnfree} all monomials of the
 form $\epsilon_i x_n$ for $i = 1,\ldots,D$ are in $B$ or in
 $\stair{g \cdot \id{I}}$ and their normal are easily known.
 Moreover, as mentioned in Section~\ref{sec:notations}, there exists
 $U'$ a the Zariski open subset of $\text{GL}(\K,n)$ such that for all
 $g \in U'$ the ideal $g \cdot \id{I}$ admits a LEX Gröbner basis
 in \textit{Shape Position}.  If $g$ is also chosen in $U'$ then we
 can use Algorithm~\ref{algchgord} to compute the LEX Gröbner basis of
 $g \cdot \id{I}$.  Hence, we propose in Algorithm~\ref{new_posso} a
 Las Vegas algorithm to solve the PoSSo problem. A Las Vegas algorithm
 is a randomized algorithm whose output (which can be \emph{fail}) is
 always correct.  The end of this section is devoted to evaluate its
 complexity and its probability of failure \textit{i.e.} when the
 algorithm returns \emph{fail}.

\begin{algorithm}[!ht]
\SetKwData{Left}{left}\SetKwData{This}{this}\SetKwData{Up}{up}
\SetKwFunction{Union}{Union}\SetKwFunction{FindCompress}{FindCompress}
\SetKwInOut{Input}{Input}\SetKwInOut{Output}{Output} \Input{A
  polynomial system $\mathcal{S} \subset \K[x_1,\ldots,x_n]$
  generating a radical ideal.}  \Output{$g$  in
  $\text{GL}(\K,n)$ and the LEX Gröbner basis of $\langle
  g \cdot \mathcal{S}\rangle$ \textit{i.e.} a univariate
  parametrization of the solutions of
  $\mathcal{S}$ or \emph{fail}.}  \caption{Another algorithm for
  PoSSo.\label{new_posso}} 
  Choose randomly $g$ in $\text{GL}(\K,n)$\;
  Compute $\gdrl$ the DRL Gröbner basis of $g \cdot \mathcal{S}$\;
  \If{$T_n$ can be read from $\gdrl$}{  
  Extract $T_n$ from $\gdrl$\; 
  From $T_n$ and $\gdrl$ compute $\glex$
  using Algorithm~\ref{algchgord}\; 
  \textbf{if} Algorithm~\ref{algchgord} succeeds \textbf{then} \Return $g$ and $\glex$\;
  \textbf{else} \Return \emph{fail}\;
}
  \textbf{else return} \emph{fail}\;
\end{algorithm}

Algorithm~\ref{new_posso} successes if the three following conditions
 are satisfied
\begin{enumerate}
\item\label{radical} $g \in \text{GL}(\K,n)$ is chosen in a non empty Zariski open set $U'$ such that for all $g \in U'$, $g \cdot \id{I}$ has a LEX Gröbner basis in \textit{Shape Position};  
\item\label{zariski} $g \in \text{GL}(\K,n)$ is chosen in a non empty Zariski open set $U$ such that for all $g \in U$, $\inId{drl}{g \cdot \id{I}} = \text{Gin}(\id{I}^h)$;
\item\label{degree} $p = 0$ or $p > 0$ and for all $m \in E(g \cdot \id{I})$, $\deg_{x_n}(m) \not \equiv 0 \mod p$.
\end{enumerate}

The existence of the non empty Zariski open subset $U'$ is proven
in \cite{GiMo89}. Conditions~\eqref{radical} and \eqref{zariski} are
satisfied if $g \in U \cap U'$. Since, $U$ and $U'$ are open and
dense, $U \cap U'$ is also a non empty Zarisky open set.

\subsubsection{Probability of failure of Algorithm~\ref{new_posso}}\label{sec:prob}

Usually, the coefficient field of the polynomials is the field of
rational numbers or a finite field.  For fields of characteristic
zero, if $g$ is chosen randomly then the probability that the
condition~\eqref{radical} and \eqref{zariski} be satisfied is 1. By
consequence, the probability of failure of Algorithm~\ref{new_posso},
in case of field of characteristic zero, is 0.

For finite fields $\F_q$, the Schwartz-Zippel lemma \cite{Sch80,Zip79}
allows to bound the probability that the conditions~\eqref{radical}
and \eqref{zariski} do not be satisfied by $\frac{d}{q}$ where $d$ is
the degree of the polynomial defining $U \cap U'$.  Thus, in order to
bound this failure probability we recall briefly how are constructed
$U$ and $U'$.

\paragraph{Construction of $U'$.}
Let $\id{I} = \Id{f_1,\ldots,f_n}$ be a radical ideal of
$\K[x_1,\ldots,x_n]$. Since $\id{I}$ is radical, all its solutions
are distinct. Therefore, let $\{ a_i =
(a_{i,1},\ldots,a_{i,n}) \in \overline{\K}^n\ |\ f_j(a_1,\ldots,a_n) =
0,\ j=1,\ldots,n\}$ be the set of solutions of $\id{I}$ (recall that
its cardinality is $D$).   Let $g$ be a given matrix in $\text{GL}(\K,n)$. We
denote by $v_i = (v_{i,1},\ldots,v_{i,n})$ the point obtained after
transformation of $a_i$ by $g$, \textit{i.e} $v_i = g \cdot a_i^t$.  To
ensure that $g \cdot \id{I}$ admits a LEX Gröbner basis
in \textit{Shape Position}, $g$ should be such that $v_{i,n} \ne
v_{j,n}$ for all couples of integers $(i,j)$ verifying $1 \leq j <
i \leq D$.  Hence, let $\mathbf{g} = (\mathbf{g}_{i,j})$ be a
$(n\times n)$ matrix of unknowns, the polynomial $P_{U'}$ defining the
Zariski open subset $U'$ is then given as the determinant of the
Vandermonde matrix associated to $\mathbf{v}_{i,n}$ for 
$i=1,\ldots,D$ where $\mathbf{v}_i =
(\mathbf{v}_{i,1},\ldots,\mathbf{v}_{i,n}) = \mathbf{g} \cdot
a_i^t$. Therefore, we know exactly the degree of $P_{U'}$ which is
$\frac{D(D-1)}{2}$.

\paragraph{Construction of $U$.}
The Zariski open subset $U$ is constructed as the intersection of
Zariski open subsets $U_1,\ldots,U_\delta$ of $\text{GL}(\K,n)$ where
$\delta$ is the maximum degree of the generators of
$\text{Gin}(\id{I}^h)$. Let $d$ be a fixed degree.  Let
$\K[x_1,\ldots,x_n]_d = R_d$ be the set of homogeneous polynomials of
degree $d$ of $\K[x_1,\ldots,x_n]$. Let $\{f_1,\ldots,f_{t_d}\}$ be
a vector basis of $\id{I}^h_d = \id{I}^h \cap R_d$.
Let $\mathbf{g} = (\mathbf{g}_{i,j})$ be a $(n\times n)$ matrix of
unknowns and let $M$ be a matrix representation of the map
$\id{I}^h_d \rightarrow\mathbf{g} \cdot \id{I}^h_d$
defined
as follow:
$$
\begin{array}{c|ccc|c}
 \multicolumn{1}{c}{} & \multicolumn{1}{c}{m_1} & \cdots & \multicolumn{1}{c}{m_N} \\  
\cline{2-4}
& \star & \cdots & \star & \mathbf{g} \cdot f_1\\
M = (M_{i,j}) = & \vdots & \ddots & \vdots & \vdots\\ 
& \star & \cdots & \star & \mathbf{g} \cdot f_{t_d}\\
\cline{2-4}
\end{array}
$$ where $M_{i,j}$ is the coefficient of $m_j$ in $\mathbf{g} \cdot
f_i$ and $\{m_1,\ldots,m_N\}$ is the set of monomials in $R_d$.
In \cite{BS87,Eis95}, the polynomial $P_{U_d}$ defining $U_d$ is
constructed as a particular minor of size $t_d$ of $M$. Since each
coefficient in $M$ is a polynomial in
$\K[\mathbf{g}_{1,1},\ldots,\mathbf{g}_{n,n}]$ of degree $d$, the
degree of $P_{U_d}$ is $d\cdot t_d$.  Finally, since $U_d$ is open and
dense for all $d=1,\ldots,\delta$ we deduce that $U
= \cap_{i=1}^\delta U_d$ is a non empty Zariski open set whose
defining polynomial, $P_U$, is of degree $\sum_{d=1}^\delta d \cdot
t_d \leq \delta \sum_{i=1}^\delta t_d$. Moreover, $D
= \operatorname{dim}_\K(\K[x_1,\ldots,x_n]/\id{I}^h)
= \sum_{d=0}^\delta \operatorname{dim}_\K(R_d/\id{I}^h_d)$. Thus,
$\sum_{d=0}^\delta \operatorname{dim}_\K(\id{I}^h_d)
= \sum_{d=0}^\delta \operatorname{dim}_\K(R_d) - D
= \binom{n+\delta}{n}-D$.  By consequence,
$\deg(P_U) \leq \delta \left ( \binom{n+\delta}{n} - D \right )$.

\paragraph{}
For ideals generated by a regular sequence $(f_1,\ldots,f_n)$, thanks
to the Macaulay's bound, $\delta$ can be bounded by
$\sum_{i=1}^n(\deg(f_i)-1)+1$. Note that the Macaulay's bound gives
also a bound on $\deg_{x_n}(m)$ for all $m \in E(g \cdot \id{I})$.  To
conclude, if $p > \sum_{i=1}^n(\deg(f_i)-1)+1$ then
condition~\eqref{degree} is satisfied and for any $p$ the probability
that conditions~\eqref{radical} and \eqref{zariski} be satisfied is
greater than $$1 - \frac{1}{q}\left ( \frac{D(D-1)}{2} + \left
(\sum_{i=1}^n(\deg(f_i)-1)+1\right )\left
( \binom{\sum_{i=1}^n\deg(f_i)+1}{n} - D \right )\right )\,.$$

\subsubsection{Complexity of Algorithm~\ref{new_posso}}

As previously mentioned, the matrix $T_n$ can be read from $\gdrl$
(test in Line~3 of Algorithm~\ref{new_posso}) if all the monomials of
the form $\epsilon_i x_n$ are either in $B$ or in $E(\Id{\gdrl})$. Let
$F_n = \{\epsilon_i x_n\ |\ i = 1,\ldots,D\}$, the test in Line~3 is
equivalent to test if $F_n \subset B \cup \stair{\Id{\gdrl}}$. Since
$F_n$ contains exactly $D$ monomials and $B \cup \stair{\Id{\gdrl}}$
contains at most $(n+1)D$ monomials; in a similar way as in
Lemma~\ref{lem:BF} testing if $F_n \subset B \cup \stair{\Id{\gdrl}}$
can be done in at most $O(nD^2)$ elementary operations which can be
decreased to $O(D)$ elementary operations if we use a hash table.
Hence, the cost of computing $B$, $F_n$ (see Lemma~\ref{lem:BF}) and
the test in Line~3 of Algorithm~\ref{new_posso} are negligible in
comparison to the complexity of Algorithm~\ref{algchgord}. Hence, the
complexity of Algorithm~\ref{new_posso} is given by the complexity of
$F_5$ algorithm to compute the DRL Gröbner basis of $g \cdot \id{I}$
and the complexity of Algorithm~\ref{algchgord} to compute the LEX
Gröbner basis of $g \cdot \id{I}$.  From \cite{Laz83}, the
complexities of computing the DRL Gröbner basis of $g \cdot \id{I}$ or
$\id{I}$ are the same.  Since it is straightforward to see that the
number of solutions of these two ideals are also the same we obtain
the second main result of the paper.

\begin{theorem}\label{thm:main2}
Let \(\K\) be the rational field $\Q$ or a finite field $\F_q$ of sufficiently large characteric $p$. Let \(\mathcal{S} = \{f_1,\ldots,f_n\}\subset
\K[x_1,\ldots,x_n]\) be a polynomial system generating a radical ideal $\id{I} = \langle\mathcal{S}\rangle$ of degree \(D\). If the sequence $(f_1,\ldots,f_n)$ is a
regular sequence such that the degree of each polynomial is uniformly
bounded by a fixed or non fixed parameter $d$ then there exists a Las
Vegas algorithm which computes the univariate polynomial
representation of the roots of \(\mathcal{S}\) in $O(nd^{\omega n} +
\log_2(D)(D^\omega+n\log_2(D)D))$ arithmetic operations.
\end{theorem}

As previously mentioned, the Bézout bound allows to bound the number
of solutions $D$ by the product of the degrees of the input
equations. Since this bound is generically reached we get the
following corollary.

\begin{corollary}
Let \(\K\) be the rational field $\Q$ or a finite field $\F_q$ of sufficiently large characteric $p$. Let \(\mathcal{S} = \{f_1,\ldots,f_n\}\subset
\K[x_1,\ldots,x_n]\) be a generic polynomial system generating an ideal $\id{I} = \langle\mathcal{S}\rangle$ of degree \(D\). If the degree of each polynomial in $\mathcal{S}$ is uniformly
bounded by a fixed or non fixed parameter $d$ then there exists a Las
Vegas algorithm which computes the univariate polynomial
representation of the roots of \(\mathcal{S}\) in
$\widetilde{O}(D^\omega)=\widetilde{O}(d^{\omega n})$ arithmetic
operations where the notation $\widetilde{O}$ means that we neglect
logarithmic factors in $D$ and polynomial factors in $n$.
\end{corollary}

\section{Acknowledgments}

The authors would like to thanks André Galligo and Daniel Lazard for
fruitful discussions about generic initial ideals and Vanessa Vitse
for providing us with an example for which the randomization is
required to \textit{``freely''} construct the multiplication matrix
$T_n$.

\bibliographystyle{abbrv}
\bibliography{orderChange}

\begin{thebibliography}{10}

\bibitem{BFSY05}
M.~Bardet, J.-C. Faugère, B.~Salvy, and B.~Yang.
\newblock Asymptotic behaviour of the degree of regularity of semi-regular
  polynomial systems.
\newblock In P.~Gianni, editor, {\em The Effective Methods in Algebraic
  Geometry Conference, Mega 2005}, pages 1 --14, May 2005.

\bibitem{roadmap2}
S.~Basu, M.-F. Roy, M.~{Safey El Din}, and {\'E}.~Schost.
\newblock A baby step-giant step roadmap algorithm for general algebraic sets.
\newblock {\em CoRR}, abs/1201.6439, 2012.

\bibitem{BaSt83}
W.~Baur and V.~Strassen.
\newblock The complexity of partial derivatives.
\newblock {\em Theoretical computer science}, 22(3):317--330, 1983.

\bibitem{BaSt87}
D.~Bayer and M.~Stillman.
\newblock A criterion for detecting m-regularity.
\newblock {\em Inventiones mathematicae}, 87(1):1--11, 1987.

\bibitem{BS87}
D.~Bayer and M.~Stillman.
\newblock A theorem on refining division orders by the reverse lexicographic
  order.
\newblock {\em Duke Mathematical Journal}, 55(2):321--328, 1987.

\bibitem{shapepos}
E.~Becker, T.~Mora, M.~G. Marinari, and C.~Traverso.
\newblock The shape of the shape lemma.
\newblock In {\em Proceedings of the international symposium on Symbolic and
  algebraic computation}, ISSAC '94, pages 129--133, New York, NY, USA, 1994.
  ACM.

\bibitem{BSS03}
A.~Bostan, B.~Salvy, and E.~Schost.
\newblock Fast algorithms for zero-dimensional polynomial systems using
  duality.
\newblock {\em Applicable Algebra in Engineering, Communication and Computing},
  14(4):239--272, 2003.

\bibitem{BGY80}
R.~P. Brent, F.~G. Gustavson, and D.~Y. Yun.
\newblock Fast solution of {Toeplitz} systems of equations and computation of
  {Padé} approximants.
\newblock {\em Journal of Algorithms}, 1(3):259--295, 1980.

\bibitem{BPW06b}
J.~Buchmann, A.~Pyshkin, and R.-P. Weinmann.
\newblock Block ciphers sensitive to {Gr{\"o}bner} basis attacks.
\newblock In {\em CT-RSA}, pages 313--331, 2006.

\bibitem{BPW06}
J.~Buchmann, A.~Pyshkin, and R.-P. Weinmann.
\newblock A zero-dimensional {Gröbner} basis for {AES}-128.
\newblock In {\em Fast Software Encryption}, pages 78--88. Springer, 2006.

\bibitem{BuHo74}
J.~Bunch and J.~Hopcroft.
\newblock Triangular factorization and inversion by fast matrix multiplication.
\newblock {\em Mathematics of Computation}, 28(125):231--236, 1974.

\bibitem{canny}
J.~F. Canny.
\newblock Computing roadmaps of general semi-algebraic sets.
\newblock {\em Comput. J.}, 36(5):504--514, 1993.

\bibitem{CLO07}
D.~Cox, J.~Little, and D.~O'Shea.
\newblock {\em Ideals, Varieties, and Algorithms: an Introduction to
  Computational Algebraic Geometry and Commutative Algebra}, volume~10.
\newblock Springer, 2007.

\bibitem{Dat03}
R.~S. Datta.
\newblock Universality of {N}ash equilibria.
\newblock {\em Mathematics of Operations Research}, 28(3):424--432, 2003.

\bibitem{BoPe99}
M.~De~Boer and R.~Pellikaan.
\newblock {\em Some Tapas of Computer Algebra}, volume~4, chapter {Gröbner}
  Bases for Codes.
\newblock Springer, 1999.

\bibitem{Eis95}
D.~Eisenbud.
\newblock {\em Commutative Algebra with a View Toward Algebraic Geometry}.
\newblock Springer, 1995.

\bibitem{Fau99}
J.-C. Faugère.
\newblock A new efficient algorithm for computing {Gröbner} bases ({F4}).
\newblock {\em Journal of Pure and Applied Algebra}, 139(1--3):61--88, June
  1999.

\bibitem{Fau02}
J.-C. Faugère.
\newblock A new efficient algorithm for computing {Gröbner} bases without
  reduction to zero ({F5}).
\newblock In {\em Proceedings of the 2002 international symposium on Symbolic
  and algebraic computation}, ISSAC '02, pages 75--83, New York, NY, USA, 2002.
  ACM.

\bibitem{FGHR12}
J.-C. Faugère, P.~Gaudry, L.~Huot, and G.~Renault.
\newblock Fast change of ordering with exponent {$\omega$}.
\newblock {\em ACM Commun. Comput. Algebra}, 46:92--93, September 2012.

\bibitem{FGHR13}
J.-C. Faugère, P.~Gaudry, L.~Huot, and G.~Renault.
\newblock Using symmetries in the index calculus for elliptic curves discrete
  logarithm.
\newblock {\em To appear in Journal of Crypotlogy}, 2013.
\newblock \url{http://eprint.iacr.org/}.

\bibitem{FGLM93}
J.-C. Faugère, P.~Gianni, D.~Lazard, and T.~Mora.
\newblock Efficient computation of zero-dimensional {Gröbner} bases by change
  of ordering.
\newblock {\em Journal of Symbolic Computation}, 16(4):329--344, 1993.

\bibitem{SparseFGLM}
J.-C. Faugère and C.~Mou.
\newblock Sparse {FGLM} algorithms.
\newblock \url{http://hal.inria.fr/hal-00807540}.

\bibitem{FaMo11}
J.-C. Faugère and C.~Mou.
\newblock Fast algorithm for change of ordering of zero-dimensional {Gröbner}
  bases with sparse multiplication matrices.
\newblock In {\em Proceedings of the 36th international symposium on Symbolic
  and algebraic computation}, ISSAC '11, pages 115--122, New York, NY, USA,
  2011. ACM.

\bibitem{Gal73}
A.~Galligo.
\newblock {\em A Propos du Théorème de Préparation de {Weierstrass}}.
\newblock PhD thesis, Institut de Mathématique et Sciences Physiques de
  l'Université de Nice, 1973.

\bibitem{Gau09}
P.~Gaudry.
\newblock Index calculus for abelian varieties of small dimension and the
  elliptic curve discrete logarithm problem.
\newblock {\em {Journal of Symbolic Computation}}, 44(12):1690--1702, 2009.

\bibitem{GiMo89}
P.~Gianni and T.~Mora.
\newblock Algebraic solution of systems of polynomial equations using
  {Gröbner} bases.
\newblock In {\em Applied Algebra, Algebraic Algorithms and Error Correcting
  Codes, Proceedings of AAECC-5, volume 356 of LNCS}, pages 247--257. Springer,
  1989.

\bibitem{GrSa11}
A.~Greuet and M.~Safey El~Din.
\newblock Deciding reachability of the infimum of a multivariate polynomial.
\newblock In {\em I{SSAC} 2011---{P}roceedings of the 36th {I}nternational
  {S}ymposium on {S}ymbolic and {A}lgebraic {C}omputation}, pages 131--138.
  ACM, New York, 2011.

\bibitem{JoMa89}
E.~Jonckheere and C.~Ma.
\newblock A simple {Hankel} interpretation of the {Berlekamp-Massey} algorithm.
\newblock {\em Linear Algebra and its Applications}, 125:65--76, 1989.

\bibitem{Jou13}
A.~Joux.
\newblock A new index calculus algorithm with complexity ${L}(1/4+o(1))$ in
  very small characteristic.
\newblock Cryptology ePrint Archive, Report 2013/095, 2013.
\newblock \url{http://eprint.iacr.org/}.

\bibitem{Kel85}
W.~Keller-Gehrig.
\newblock Fast algorithms for the characteristic polynomial.
\newblock {\em Theor. Comput. Sci.}, 36:309--317, June 1985.

\bibitem{KobFuFu}
H.~Kobayashi, T.~Fujise, and A.~Furukawas.
\newblock Solving systems of algebraic equations by a general elimination
  method.
\newblock {\em Journal of Symbolic Computation}, 5(3):303 -- 320, 1988.

\bibitem{Lak90}
Y.~N. Lakshman.
\newblock On the complexity of computing a {Gröbner} basis for the radical of
  a zero dimensional ideal.
\newblock In {\em Proceedings of the twenty-second annual ACM symposium on
  Theory of computing}, STOC '90, pages 555--563, New York, NY, USA, 1990. ACM.

\bibitem{LaLa91}
Y.~N. Lakshman and D.~Lazard.
\newblock On the complexity of zero-dimensional algebraic systems.
\newblock In {\em Effective methods in algebraic geometry}, volume~94, page
  217. Birkhauser, 1991.

\bibitem{Laz83}
D.~Lazard.
\newblock Gröbner bases, {G}aussian elimination and resolution of systems of
  algebraic equations.
\newblock In J.~van Hulzen, editor, {\em Computer Algebra}, volume 162 of {\em
  Lecture Notes in Computer Science}, pages 146--156. Springer Berlin /
  Heidelberg, 1983.

\bibitem{LoYo97}
P.~Loustaunau and E.~V. York.
\newblock On the decoding of cyclic codes using {Gröbner} bases.
\newblock {\em Applicable Algebra in Engineering, Communication and Computing},
  8(6):469--483, 1997.

\bibitem{Mac16}
F.~S. Macaulay.
\newblock {\em The algebraic theory of modular systems}.
\newblock Cambridge Mathematical Library. Cambridge University Press,
  Cambridge, 1994.
\newblock Revised reprint of the 1916 original, With an introduction by Paul
  Roberts.

\bibitem{Mas69}
J.~Massey.
\newblock Shift-register synthesis and {BCH} decoding.
\newblock {\em IEEE Transactions on Information Theory}, 15(1):122--127, 1969.

\bibitem{Mor03}
G.~Moreno-Socias.
\newblock Degrevlex {Gröbner} bases of generic complete intersections.
\newblock {\em Journal of Pure and Applied Algebra}, 180(3):263--283, 2003.

\bibitem{MoPa98}
B.~Mourrain and V.~Y. Pan.
\newblock Asymptotic acceleration of solving multivariate polynomial systems of
  equations.
\newblock In {\em Proceedings of the thirtieth annual ACM symposium on Theory
  of computing}, pages 488--496. ACM, 1998.

\bibitem{Pan02}
V.~Y. Pan.
\newblock Univariate polynomials: Nearly optimal algorithms for numerical
  factorization and root-finding.
\newblock {\em Journal of Symbolic Computation}, 33(5):701 -- 733, 2002.

\bibitem{Par94}
K.~Pardue.
\newblock {\em Nonstandard {Borel}-Fixed Ideals}.
\newblock PhD thesis, Brandeis University, 1994.

\bibitem{Sch80}
J.~T. Schwartz.
\newblock Fast probabilistic algorithms for verification of polynomial
  identities.
\newblock {\em J. ACM}, 27(4):701--717, Oct. 1980.

\bibitem{Stu02}
B.~Sturmfels.
\newblock {\em Solving Systems of Polynomial Equations}, volume~97.
\newblock American Mathematical Society, 2002.

\bibitem{Vas12}
V.~Vassilevska~Williams.
\newblock Multiplying matrices faster than {Coppersmith-Winograd}.
\newblock In {\em Proceedings of the 44th symposium on Theory of Computing},
  pages 887--898. ACM, 2012.

\bibitem{GaGe03}
J.~Von Zur~Gathen and J.~Gerhard.
\newblock {\em Modern Computer Algebra}.
\newblock Cambridge University Press, 2003.

\bibitem{Wie86}
D.~Wiedemann.
\newblock Solving sparse linear equations over finite fields.
\newblock {\em IEEE Transactions on Information Theory}, 32(1):54--62, 1986.

\bibitem{Zip79}
R.~Zippel.
\newblock Probabilistic algorithms for sparse polynomials.
\newblock In E.~Ng, editor, {\em Symbolic and Algebraic Computation}, volume~72
  of {\em Lecture Notes in Computer Science}, pages 216--226. Springer Berlin
  Heidelberg, 1979.

\end{thebibliography}

\newpage

\appendix

\section{Impact of Algorithm~\ref{new_posso} on the practical resolution of the PoSSo problem in the worst case}\label{appendix:bench}


In this appendix we discuss about the impact of
Algorithm~\ref{new_posso} on the practical resolution of the PoSSo
problem. Note that Algorithm~\ref{algchgord} to compute the LEX
Gröbner basis given the multiplication matrix $T_n$ is of theoretical
interest. Hence, in practice we use the sparse version of Faugère and
Mou \cite{FaMo11}. In Table~\ref{tab}, we give the time to compute the
LEX Gröbner basis using the usual algorithm
(Algorithm~\ref{usual_posso}) and Algorithm~\ref{new_posso}. This time
is divided into three steps, the first is the time to compute the DRL
Gröbner basis using $F_5$ algorithm, the second is the time to compute
the multiplication matrix $T_n$ and the last part is the time to
compute the LEX Gröbner basis given $T_n$ using the algorithm
in \cite{FaMo11}. Since, this algorithm takes advantage of the
sparsity of the matrix $T_n$ we also give its density.  We also give
the number of normal forms to compute (\textit{i.e.}  the number of
terms of the form $\epsilon_i x_n$ that are not in $B$ or in
$\stair{\id{I}}$ (or in $\stair{g \cdot \id{I}}$).

The experiments are performed on a worst case for our algorithm in the
sense that the system in input is already a DRL Gröbner basis. Thus,
while the usual algorithm does not have to compute the DRL Gröbner
basis, our algorithm need to compute the DRL Gröbner basis of
$g \cdot \id{I}$. The system in input is of the form $\mathcal{S}
= \{f_1,\ldots,f_n\} \subset \F_{65521}[x_1,\ldots,x_n]$ with
$\LT{drl}{f_i} = x_i^2$. Hence, the monomials in the basis $B$ are all
the monomials of degree at most one in each variable. The degree of
the ideal $D$ is then $2^n$. The monomials $\epsilon_i x_n$ that are
not in $B$ or in $\stair{\Id{\mathcal{S}}}$ are of the form $x_n^2m$
where $m$ is a monomial in $x_1,\ldots,x_{n-1}$ of total degree
greater than zero and linear in each variable. By consequence, using
the usual algorithm we have to compute $2^{n-1}-1$ normal forms to
compute only $T_n$.

\begin{table}[!h]
\centering
\begin{tabular}{|c|c|c|c|c|c|c|c|c|}
\hline
\multirow{2}{*}{n} & \multirow{2}{*}{D} & \multirow{2}{*}{Algorithm} & First & Build & \multirow{2}{*}{$\#$ NF} & \multirow{2}{*}{Density} &
Compute & Total\\
& & & GB & $T_n$ & & & $h_1,\ldots,h_n$ & PoSSo \\
\hline
\multirow{2}{*}{7} & \multirow{2}{*}{128} & usual & 0s & 0s & 63 & 34.20\% & 0s & 0s\\
& & This work & 0s & 0s & 0 & 26.57\% & 0s & 0s\\
\hline
\multirow{2}{*}{9} & \multirow{2}{*}{512} & usual & 0s & 13s & 255 & 32.81\% & 0s & 13s\\
& & This work & 0s & 0s & 0 & 23.68\% & 0s & 0s\\
\hline
\multirow{2}{*}{11} & \multirow{2}{*}{2048} & usual & 0s & 7521s & 1023 & 31.93\% & 23s & 7544s \\
& & This work & 5s & 0s & 0 & 21.53\% & 0s & 5s\\
\hline
\multirow{2}{*}{13} & \multirow{2}{*}{8192} & usual & 0s & $> 2$ days & 4095 & & & $> 2$ days \\
& & This work & 157s & 2s & 0 & 19.86\% & 26s & 185s\\
\hline
\multirow{2}{*}{15} & \multirow{2}{*}{32768} & usual & 0s & $> 2$ days & 16383 & & & $> 2$ days \\
& & This work & 5786s & 46s & 0 & 18.52\% & 1886s & 7718s\\
\hline
\multirow{2}{*}{16} & \multirow{2}{*}{65536} & usual & 0s & $> 2$ days & 32767 & & & $> 2$ days \\
& & This work & 38067s & 195s & 0 & 18.33\% & 14297s & 52559s\\
\hline
\end{tabular}
\caption{A worst case example: comparison of the usual algorithm for solving the PoSSo problem and Algorithm~\ref{new_posso}, the proposed algorithm. Computation with FGb on a 3.47 GHz Intel Xeon X5677 CPU.\label{tab}}
\end{table}

One can note that in the usual algorithm the bottleneck of the
resolution of the PoSSo problem is the change of ordering due to the
construction of the multiplication matrix $T_n$. Since our algorithm
allows to compute very efficiently the matrix $T_n$ (for instance for
$n=11$, 0 seconds in comparison to 7544 seconds for the usual
algorithm), the most time consuming step becomes the computation of
the DRL Gröbner basis. However, the total running time of our
algorithm is far less than that of the usual algorithm. For instance,
for $n=13$ the PoSSo problem can now be solved in approximately three
minutes whereas we are not allow to solve this instance of the PoSSo
problem using the usual algorithm.

Moreover, using Algorithm~\ref{new_posso} the density of the
matrix $T_n$ is decreased (which implies that the running time of
Faugère and Mou algorithms is also decreased). This can be explained
by the fact that the dense columns of the matrix $T_n$ comes from
monomials of the form $x_n\epsilon_i$ that are not in
$B$ \textit{i.e.} in the frontier. Since
Algorithm~\ref{new_posso} allows to ensure that the monomials
$x_n\epsilon_i$ are either in $B$ or in $\stair{g \cdot \id{I}}$ then
the number of dense columns in $T_n$ is potentially decreased.

\end{document}